\def\draft{0}
\documentclass[11pt,letterpaper]{article}
\usepackage[margin=1in]{geometry}
\usepackage[utf8]{inputenc}
\usepackage{amsthm}
\usepackage{amsthm, amsmath, amssymb}
\usepackage{bm}
\usepackage{graphicx}
\usepackage{color,xcolor}
\usepackage{enumitem}
\usepackage[normalem]{ulem}

\usepackage[colorlinks=true, allcolors=blue]{hyperref}
\usepackage[capitalise,nameinlink]{cleveref}
\usepackage[style=alphabetic, backend=biber, minalphanames=3, maxalphanames=4, maxbibnames=99]{biblatex}

\newcommand{\vnote}[1]{\ifnum\draft=1 {\color{green!70!black} [\textbf{SV:} #1]}\fi}
\newcommand{\mnote}[1]{\ifnum\draft=1 {\color{red} [\textbf{MS:} #1]}\fi}
\newcommand{\snote}[1]{\ifnum\draft=1 {\color{teal} [\textbf{NS:} #1]}\fi}
\newcommand{\hnote}[1]{\ifnum\draft=1 {\color{magenta} [\textbf{MH:} #1]}\fi}
\newcommand{\bnote}[1]{\ifnum\draft=1 {\color{orange} [\textbf{JB:} #1]}\fi}
\newcommand{\pnote}[1]{\ifnum\draft=1 {\color{brown} [\textbf{TMP:} #1]}\fi}

\newcommand{\gbmaxf}{{(\gamma,\beta)\textrm{-}\maxf}}

\newcommand{\Exp}{\mathop{\mathbb{E}}}

\newcommand{\cA}{\mathcal{A}}

\newcommand{\cD}{\mathcal{D}}

\newcommand{\N}{\mathbb{N}}

\newcommand{\R}{\mathbb{R}}
\newcommand{\bern}{\mathsf{Bern}}

\newcommand{\m}[1][]{\textsf{Max-#1}}
\newcommand{\mcsp}{\textsf{Max-CSP}}

\newcommand{\sym}{\mathsf{S}}

\newcommand{\bias}{\textsf{bias}}
\newcommand{\val}{\textsf{val}}

\newcommand{\supp}{\textsf{supp}}

\newcommand{\maxf}{\textsf{Max-CSP}(f)}
\newcommand{\dmaxf}{\textsf{-Max-CSP}(f)}
\newcommand{\twoand}{\textsf{$2$AND}}
\newcommand{\maxtwoand}{\textsf{Max-$2$AND}}
\newcommand{\threeand}{\textsf{$3$AND}}
\newcommand{\maxthreeand}{\textsf{Max-$3$AND}}

\newcommand{\kand}{k\textsf{AND}}
\newcommand{\maxkand}{\textsf{Max-}\kand}

\newcommand{\NP}{\mathbf{NP}}

\newcommand{\veca}{\mathbf{a}}
\newcommand{\vecb}{\mathbf{b}}

\newcommand{\vecj}{\mathbf{j}}

\newcommand{\vecx}{\mathbf{x}}
\newcommand{\vecy}{\mathbf{y}}

\newcommand{\vecsigma}{\boldsymbol{\sigma}}

\newcommand{\vecmu}{\boldsymbol{\mu}}

\newcommand{\vecpi}{\boldsymbol{\pi}}

\renewcommand{\tilde}{\widetilde}
\renewcommand{\hat}{\widehat}

\newcommand{\wt}{\operatorname{wt}}
\newcommand{\rroot}{\operatorname{root}_{\R}}

\newcommand\eqdef{\stackrel{\mathrm{\small def}}{=}}

\newcommand{\kz}{\{0\}\cup[k]}

\newcommand{\diff}{\mathsf{diff}}
\newcommand{\vecopt}{\mathbf{opt}}
\newcommand{\vecmaj}{\mathbf{maj}}

\newcommand{\Th}{\mathsf{Th}}
\newcommand{\Ex}{\mathsf{Ex}}
\newcommand{\PSPACE}{\mathsf{PSPACE}}

\newcommand{\dsym}{\mathcal{D}^{\mathrm{sym}}}

\newcommand{\Unif}{\mathsf{Unif}}
\usepackage{amsthm}
\usepackage{thmtools,thm-restate}

\numberwithin{equation}{section}
\declaretheoremstyle[bodyfont=\it,qed=\qedsymbol]{noproofstyle}

\declaretheorem[numberlike=equation]{observation}

\declaretheorem[name=Observation,numbered=no]{observation*}

\declaretheorem[numberlike=equation]{theorem}

\declaretheorem[name=Theorem,numbered=no]{theorem*}

\declaretheorem[numberlike=equation]{lemma}
\declaretheorem[name=Lemma,numbered=no]{lemma*}

\declaretheorem[numberlike=equation]{corollary}
\declaretheorem[name=Corollary,numbered=no]{corollary*}

\declaretheorem[numberlike=equation]{proposition}
\declaretheorem[name=Proposition,numbered=no]{proposition*}

\declaretheorem[name=Claim,numbered=no]{claim*}

\declaretheorem[name=Conjecture,numbered=no]{conjecture*}

\declaretheorem[name=Question,numbered=no]{question*}

\declaretheoremstyle[bodyfont=\it]{defstyle} 

\declaretheorem[numberlike=equation,style=defstyle]{definition}
\declaretheorem[unnumbered,name=Definition,style=defstyle]{definition*}

\declaretheorem[unnumbered,name=Example,style=defstyle]{example*}

\declaretheorem[unnumbered,name=Notation=defstyle]{notation*}

\declaretheorem[unnumbered,name=Construction,style=defstyle]{construction*}

\declaretheoremstyle[]{rmkstyle} 

\newtheorem*{remark}{Remark}

\title{On sketching approximations for symmetric Boolean CSPs}
\author{Joanna Boyland\thanks{Harvard College, Harvard University, Cambridge, MA, USA. Emails: \texttt{\{jboyland,michaelhwang1, tarunmuraliprasad,noahsinger\}@college.harvard.edu}.}
\and Michael Hwang\footnotemark[1]
\and Tarun Prasad\footnotemark[1]
\and Noah Singer\footnotemark[1] \thanks{Supported by the Harvard College Research Program.}
\and Santhoshini Velusamy\thanks{School of Engineering and Applied Sciences, Harvard University, Cambridge, MA, USA. Supported in part by a Google Ph.D. Fellowship, a Simons Investigator Award to Madhu Sudan, and NSF Awards CCF 1715187 and CCF 2152413. Email: \texttt{svelusamy@g.harvard.edu}.}
}

\begin{document}

\maketitle

\begin{abstract}
    A Boolean maximum constraint satisfaction problem, $\maxf$, is specified by a predicate $f:\{-1,1\}^k\to\{0,1\}$. An $n$-variable instance of $\maxf$ consists of a list of constraints, each of which applies $f$ to $k$ distinct literals drawn from the $n$ variables. For $k=2$, Chou, Golovnev, and Velusamy \cite{CGV20} obtained explicit ratios characterizing the $\sqrt n$-space \emph{streaming} approximability of every predicate. For $k \geq 3$, Chou, Golovnev, Sudan, and Velusamy \cite{CGSV21-boolean} proved a general dichotomy theorem for $\sqrt n$-space \emph{sketching} algorithms: For every $f$, there exists $\alpha(f)\in (0,1]$ such that for every $\epsilon>0$, $\maxf$ is $(\alpha(f)-\epsilon)$-approximable by an $O(\log n)$-space \emph{linear} sketching algorithm, but $(\alpha(f)+\epsilon)$-approximation sketching algorithms require $\Omega(\sqrt{n})$ space.
     
    In this work, we give closed-form expressions for the sketching approximation ratios of multiple families of \emph{symmetric} Boolean functions. Letting $\alpha'_k = 2^{-(k-1)} (1-k^{-2})^{(k-1)/2}$, we show that for odd $k \geq 3$, $\alpha(\kand) = \alpha'_k$, and for even $k \geq 2$, $\alpha(\kand) = 2\alpha'_{k+1}$. Thus, for every $k$, $\kand$ can be $(2-o(1))2^{-k}$-approximated by $O(\log n)$-space \emph{sketching} algorithms; we contrast this with a lower bound of Chou, Golovnev, Sudan, Velingker, and Velusamy~\cite{CGS+22} implying that \emph{streaming} $(2+\epsilon)\cdot2^{-k}$-approximations require $\Omega(n)$ space! We also resolve the ratio for the ``at-least-$(k-1)$-$1$'s'' function for all even $k$; the ``exactly-$\frac{k+1}2$-$1$'s'' function for odd $k \in \{3,\ldots,51\}$; and fifteen other functions. We stress here that for general $f$, the dichotomy theorem in \cite{CGSV21-boolean} only implies that $\alpha(f)$ can be computed to arbitrary precision in $\PSPACE$, and thus closed-form expressions need not have existed \emph{a priori}. Our analyses involve identifying and exploiting structural ``saddle-point'' properties of this dichotomy.
    
    Separately, for all \emph{threshold} functions, we give optimal ``bias-based'' approximation algorithms generalizing \cite{CGV20} while simplifying \cite{CGSV21-boolean}. Finally, we investigate the $\sqrt n$-space \emph{streaming} lower bounds in \cite{CGSV21-boolean}, and show that they are \emph{incomplete} for $\threeand$, i.e., they fail to rule out $(\alpha(\threeand)-\epsilon)$-approximations in $o(\sqrt n)$ space.
\end{abstract}

\newpage

\section{Introduction}

In this work, we consider the \emph{streaming approximability} of various \emph{Boolean constraint satisfaction problems}, and we begin by defining these terms. See \cite[\S1.1-2]{CGSV21-boolean} for more details on the definitions.

\subsection{Setup: The streaming approximability of Boolean CSPs}\label{sec:setup}

\subsubsection{Boolean CSPs}\label{sec:boolean-csps}
Let $f : \{-1,1\}^k \to \{0,1\}$ be a Boolean function. In an $n$-variable instance of the problem $\mcsp(f)$, a \emph{constraint} is a pair $C = (\vecb,\vecj)$, where $\vecj = (j_1,\ldots,j_k) \in [n]^k$ is a $k$-tuple of distinct \emph{indices}, and $\vecb = (b_1,\ldots,b_k) \in \{-1,1\}^k$ is a \emph{negation pattern}.

For Boolean vectors $\veca = (a_1,\ldots,a_n),\vecb=(b_1,\ldots,b_n)\in\{-1,1\}^n$, let $\veca \odot \vecb$ denote their \emph{coordinate-wise product} $(a_1b_1,\ldots,a_nb_n)$. An \emph{assignment} $\vecsigma = (\sigma_1,\ldots,\sigma_n) \in \{-1,1\}^n$ satisfies $C$ iff $f(\vecb \odot \vecsigma|_{\vecj}) = 1$, where $\vecsigma|_{\vecj}$ is the $k$-tuple $(\sigma_{j_1},\ldots,\sigma_{j_k})$ (i.e., $\vecsigma$ satisfies $C$ iff $f(b_1\sigma_{j_1},\ldots,b_k\sigma_{j_k})=1$). An \emph{instance} $\Psi$ of $\mcsp(f)$ consists of constraints $C_1,\dots,C_m$ with non-negative weights $w_1,\ldots,w_m$ where $C_i=(\vecj(i),\vecb(i))$ and $w_i \in \R$ for each $i\in[m]$; the \emph{value} $\val_\Psi(\vecsigma)$ of an assignment $\vecsigma$ to $\Psi$ is the (weighted) fraction of constraints in $\Psi$ satisfied by $\vecsigma$, i.e., $\val_\Psi(\vecsigma)\eqdef\tfrac{1}{W}\sum_{i\in[m]}w_i \cdot f(\vecb(i)\odot\vecsigma|_{\vecj(i)})$, where $W = \sum_{i=1}^m w_i$. The \emph{value} $\val_\Psi$ of an instance $\Psi$ is the maximum value of any assignment $\vecsigma \in \{-1,1\}^n$, i.e., $\val_\Psi\eqdef\max_{\vecsigma\in\{-1,1\}^n}\val_\Psi(\vecsigma)$.

\subsubsection{Approximations to CSPs}
For $\alpha \in [0,1]$, we consider the problem of \emph{$\alpha$-approximating $\mcsp(f)$}. In this problem, the goal of an algorithm $\cA$ is to, on input an instance $\Psi$, output an estimate $\cA(\Psi)$ such that with probability at least $\frac23$, $\alpha \cdot \val_\Psi \leq \cA(\Psi) \leq \val_\Psi$. For $
\beta < \gamma \in [0,1]$, we also consider the closely related $(\beta,\gamma)\dmaxf$. In this problem, the input instance $\Psi$ is promised to either satisfy $\val_\Psi \leq \beta$ or $\val_\Psi \geq \gamma$, and the goal is to decide which is the case with probability at least $\frac23$.

\subsubsection{Streaming and sketching algorithms for CSPs}
For various Boolean functions $f$, we consider algorithms which attempt to approximate $\mcsp(f)$ instances in the \emph{(single-pass, insertion-only) space-$s$ streaming setting}. Such algorithms can only use space $s$ (which is ideally small, such as $O(\log n)$, where $n$ is the number of variables in an input instance), and, when given as input a CSP instance $\Psi$, can only read the list of constraints in a single, left-to-right pass.

We also consider a (seemingly) weak class of streaming algorithms called \emph{sketching algorithms}, where the algorithm's output is determined by an length-$s$ string called a ``sketch'' produced from the input stream, and the sketch itself has the property that the sketch of the concatenation of two streams can be computed from the sketches of the two component streams. (See \cite[\S3.3]{CGSV21-boolean} for a formal definition.) A special case of sketching algorithms are \emph{linear sketches}, where each sketch (i.e., element of $\{0,1\}^s$) encodes an element of a vector space and we perform vector addition to combine two sketches.

\subsection{Prior work and motivations}

\subsubsection{Prior results on streaming and sketching $\mcsp(f)$}

We first give a brief review of what is already known about the streaming and sketching approximability of $\maxf$. For $f : \{-1,1\}^k \to \{0,1\}$, let $\rho(f) \eqdef \Pr_{\vecb\sim\Unif(\{-1,1\}^k)}[f(\vecb)=1]$, where $\Unif(\{-1,1\}^k)$ denotes the uniform distribution on $\{-1,1\}^k$. For every $f$, the $\maxf$ problem has a trivial $\rho(f)$-approximation algorithm given by simply outputting $\rho(f)$ since $\mathbb{E}_{\veca\sim \Unif(\{-1,1\}^n)}[\val_\Psi(\veca)] = \Pr_{\vecb\sim\Unif(\{-1,1\}^k)}[f(\vecb)=1] = \rho(f)$. We refer to a function $f$ as \emph{approximation-resistant} for some class of algorithms (e.g., streaming or sketching algorithms with some space bound) if it cannot be $(\rho(f)+\epsilon)$-approximated for any constant $\epsilon > 0$. Otherwise, we refer to $f$ as \emph{approximable} for the class of algorithms.

The first two CSPs whose $o(\sqrt n)$-space streaming approximabilities were resolved were \m[$2$XOR] and \m[$2$AND]. Kapralov, Khanna, and Sudan~\cite{KKS15} showed that \m[$2$XOR] is approximation-resistant to $o(\sqrt n)$-space streaming algorithms. Later, Chou, Golovnev, and Velusamy~\cite{CGV20}, building on earlier work of Guruswami, Velusamy, and Velingker~\cite{GVV17}, gave an $O(\log n)$-space linear sketching algorithm which $(\frac49-\epsilon)$-approximates \m[$2$AND] for every $\epsilon > 0$ and showed that $(\frac49+\epsilon)$-approximations require $\Omega(\sqrt n)$ space, even for streaming algorithms.

In two recent works \cite{CGSV21-boolean,CGSV21-finite}, Chou, Golovnev, Sudan, and Velusamy proved so-called \emph{dichotomy theorems} for sketching CSPs. In \cite{CGSV21-boolean}, they prove the dichotomy for CSPs over the Boolean alphabet with negations of variables (i.e., the setup we described in \cref{sec:boolean-csps}). In \cite{CGSV21-finite}, they extend it to the more general case of CSPs over finite alphabets.\footnote{More precisely, \cite{CGSV21-boolean} and \cite{CGSV21-finite} both consider the more general case of CSPs defined by \emph{families} of functions of a specific arity. We do not need this generality for the purposes of our paper, and therefore omit it.}

\cite{CGSV21-boolean} is most relevant for our purposes, as it concerns Boolean CSPs. For a fixed constraint function $f : \{-1,1\}^k \to \{0,1\}$, the main result in \cite{CGSV21-boolean} is the following \emph{dichotomy theorem}: For any $0 \leq \gamma < \beta \leq 1$, either

\begin{enumerate}
    \item $(\beta,\gamma)\dmaxf$ has an $O(\log n)$-space linear sketching algorithm, \emph{or}
    \item For all $\epsilon > 0$, sketching algorithms for $(\beta+\epsilon,\gamma-\epsilon)\dmaxf$ require $\Omega(\sqrt n)$ space.
\end{enumerate}

Distinguishing whether (1) or (2) applies is equivalent to deciding whether two convex polytopes (which depend on $f,\gamma,\beta$) intersect. We omit a technical statement of this criterion, and instead focus on the following corollary: there exists an $\alpha(f) \in [0,1]$ such that $\maxf$ can be $(\alpha(f)-\epsilon)$-approximated by $O(\log n)$-space linear sketches, but not $(\alpha(f)+\epsilon)$-approximated by $o(\sqrt n)$-space sketches, for all $\epsilon > 0$; furthermore, $\alpha(f)$ equals the solution to an explicit minimization problem, which we describe in \cref{sec:cgsv-framework} (in the special case where $f$ is symmetric).

\emph{A priori}, it may be possible to achieve an $(\alpha(f)+\epsilon)$-approximation with a $o(\sqrt n)$-space \emph{streaming} algorithm. But \cite{CGSV21-boolean} also extends the lower bound (case 2 of the dichotomy) to cover streaming algorithms when special objects called \emph{padded one-wise pairs} exist. See \cref{sec:cgsv-streaming-framework} below for a definition (again, specialized for symmetric functions). The padded one-wise pair criterion is sufficient to recover all previous streaming approximability results for Boolean functions (i.e., \cite{KKS15,CGV20}), and prove several new ones. In particular, \cite{CGSV21-boolean} proves that if $f : \{-1,1\}^k \to \{0,1\}$ has the property that there exists $\cD \in \Delta(f^{-1}(1))$ such that $\Exp_{\vecb\sim\cD}[b_i] = 0$ for all $i \in [k]$ (where $[k] \eqdef \{1,\ldots,k\}$), then $\mcsp(f)$ is streaming approximation-resistant. For symmetric Boolean CSPs, they also prove the converse, and thus give a complete characterization for approximation resistance \cite[Lemma 2.14]{CGSV21-boolean}. However, besides \m[$2$AND], \cite{CGSV21-boolean} does not explicitly analyze the approximation ratio of any CSP that is ``approximable'', i.e., not approximation resistant.

\subsubsection{Questions from previous work}\label{sec:prev-work-qs}

In this work, we address several major questions about streaming approximations for Boolean CSPs which Chou, Golovnev, Sudan, and Velusamy~\cite{CGSV21-boolean} leave unanswered:

\begin{enumerate}
    \item Can the framework in \cite{CGSV21-boolean} be used to find closed-form sketching approximability ratios $\alpha(f)$ for approximable problems $\maxf$ beyond $\m[\twoand]$?
    \item As observed in \cite[\S1.3]{CGS+22}, \cite{CGSV21-boolean} implies the following ``trivial upper bound'' on streaming approximability: for all $f$, $\alpha(f) \leq 2\rho(f)$. How tight is this upper bound?
    \item Does the streaming lower bound (the ``padded one-wise pair'' criterion) in \cite{CGSV21-boolean} suffice to resolve the streaming approximability of every function?
    \item The optimal $(\alpha(f)-\epsilon)$-approximation algorithm for $\maxf$ in \cite{CGSV21-boolean} requires running a ``grid'' of $O(1/\epsilon^2)$ distinguishers for $(\beta,\gamma)\dmaxf$ distinguishing problems in parallel. Can we obtain simpler optimal sketching approximations?
\end{enumerate}

\subsection{Our results}

We study the questions in \cref{sec:prev-work-qs} for \emph{symmetric} Boolean CSPs. Symmetric Boolean functions are those functions that depend only on the Hamming weight of the input, i.e., number of $1$'s in the input.\footnote{Note that the inputs are in $\{-1,1\}^k$; we define the Hamming weight as the number of $1$'s, and not $-1$'s (which is arguably more ``natural'' under the mapping $b \in \{0,1\} \mapsto (-1)^b \in \{-1,1\}$), for consistency with \cite{CGSV21-boolean}.} For a set $S \subseteq [k]$, we define $f_{S,k} : \{-1,1\}^k \to \{0,1\}$ as the indicator function for the set $\{\vecb \in \{-1,1\}^k : \wt(\vecb) \in S\}$ (where $\wt(\vecb)$ denotes the Hamming weight of $\vecb$). That is, $f_{S,k}(\vecx)=1$ if and only if $\wt(\vecx)\in S$. Some well-studied examples of functions in this class include $\kand = f_{\{k\},k}$, the \emph{threshold functions} $\Th^i_k = f_{\{i,i+1,\ldots,k\},k}$, and ``exact weight'' functions $\Ex^i_k = f_{\{i\},k}$.\footnote{By \cite[Lemma 2.14]{CGSV21-boolean}, if $S$ contains elements $s \leq \frac{k}2$ and $t \geq \frac{k}2$, not necessarily distinct, then $f_{S,k}$ supports one-wise independence and is therefore approximation-resistant (even to streaming algorithms). Thus, we focus on the case where all elements of $S$ are either larger than or smaller than $\frac{k}2$. Moreover, note that if $S' = \{k-s : s \in S\}$, every instance of $\mcsp(f_{S,k})$ can be viewed as an instance of $\mcsp(f_{S',k})$ with the same value, since for any constraint $C=(\vecb,\vecj)$ and assignment $\vecsigma \in \{-1,1\}^n$, we have $f_{S,k}(\vecb\odot \vecsigma\vert_\vecj)$ = $f_{S',k}(\vecb\odot (-\vecsigma)\vert_\vecj)$. Thus, we further narrow our focus to the case where every element of $S$ is larger than $\frac{k}2$.}

\subsubsection{The sketching approximability of ${\m}\kand$}\label{sec:kand-overview}
Chou, Golovnev, and Velusamy~\cite{CGV20} showed that $\alpha(\twoand) = \frac49$ (and $(\frac49+\epsilon)$-approximation can be ruled out even for $o(\sqrt n)$-space streaming algorithms).
For $k \geq 3$, while Chou, Golovnev, Velusamy, and Sudan~\cite{CGSV21-boolean} give optimal sketching approximation algorithms for $\maxkand$, they do not explicitly analyze the approximation ratio $\alpha (\kand)$, and show only that it lies between $2^{-k}$ and $2^{-(k-1)}$.

In this paper, we analyze the dichotomy theorem in \cite{CGSV21-boolean}, and obtain a closed-form expression for the sketching approximability of $\maxkand$ for every $k$. For odd $k \geq 3$, define the constant
\begin{equation}\label{eqn:alpha'_k}
    \alpha'_k \eqdef \left(\frac{(k-1)(k+1)}{4k^2}\right)^{(k-1)/2} = 2^{-(k-1)} \cdot \left(1-\frac1{k^2}\right)^{(k-1)/2}.
\end{equation}
In \cref{sec:kand-analysis}, we prove the following:

\begin{theorem}\label{thm:kand-approximability}
For odd $k \geq 3$, $\alpha(\kand) = \alpha'_k$, and for even $k \geq 2$, $\alpha(\kand) = 2\alpha'_{k+1}$.
\end{theorem}


Since $\rho(\kand) = 2^{-k}$, \cref{thm:kand-approximability} also has the following important corollary:

\begin{corollary}\label{cor:kand-asympt}
$\lim_{k \to \infty} \frac{\alpha(\kand)}{2\rho(\kand)} = 1$.
\end{corollary}

Recall that \cite{CGSV21-boolean} implies that $\alpha(f) \leq 2\rho(f)$ for all functions $f$. Indeed, Chou, Golovnev, Sudan, Velusamy, and Velingker~\cite{CGS+22} show that any function $f$ cannot be $(2\rho(f)+\epsilon)$-approximated even by $o(n)$-space streaming algorithms. On the other hand, in \cref{sec:thresh-alg-overview} below, we describe simple $O(\log n)$-space sketching algorithms for $\maxkand$ achieving the optimal ratio from \cite{CGSV21-boolean}. Thus, as $k \to \infty$, these algorithms achieve an asymptotically optimal approximation ratio even among $o(n)$-space streaming algorithms!

\subsubsection{The sketching approximability of other symmetric functions}

We also analyze the sketching approximability of a number of other symmetric Boolean functions. Specifically, for the threshold functions $\Th^{k-1}_k$ for even $k$, we show that:

\begin{theorem}\label{thm:k-1-k-approximability}
For even $k \geq 2$, $\alpha(\Th^{k-1}_k) = \frac{k}2\alpha'_{k-1}$.
\end{theorem}

We prove \cref{thm:k-1-k-approximability} in \cref{sec:k-1-k-analysis} using techniques similar to our proof of \cref{thm:kand-approximability}. We also provide partial results for $\Ex^{(k+1)/2}_k$, including closed forms for small $k$ and an asymptotic analysis of $\alpha(\Ex^{(k+1)/2}_k)$:

\begin{theorem}[Informal version of \cref{thm:k+1/2-approximability}]\label{thm:k+1/2-approximability-informal}
For odd $k \in \{3,\ldots,51\}$, there is an explicit expression for $\alpha(\Ex^{(k+1)/2}_k)$ as a function of $k$.
\end{theorem}

\begin{theorem}\label{thm:k+1/2-asymptotic-lb}
$\lim_{\text{odd } k \to \infty} \frac{\alpha\left(\Ex^{(k+1)/2}_k\right)}{\rho\left(\Ex^{(k+1)/2}_k\right)}=1$.
\end{theorem}

We prove \cref{thm:k+1/2-approximability-informal,thm:k+1/2-asymptotic-lb} in \cref{sec:k+1/2-analysis}. Finally, in \cref{sec:other-analysis}, we explicitly resolve fifteen other cases (e.g., $f_{\{2,3\},3}$ and $f_{\{4\},5}$) not covered by \cref{thm:kand-approximability,thm:k-1-k-approximability,thm:k+1/2-approximability-informal}.

\subsubsection{Simple approximation algorithms for threshold functions}\label{sec:thresh-alg-overview}
Chou, Golovnev, and Velusamy's optimal $(\frac49-\epsilon)$-approximation for $\twoand$~\cite{CGV20}, like Guruswami, Velingker, and Velusamy's earlier $(\frac25-\epsilon)$-approximation~\cite{GVV17}, is based on measuring a quantity called the \emph{bias} of an instance $\Psi$, denoted $\bias(\Psi)$, which is defined as follows: For each $i \in [n]$, $\diff_i(\Psi)$ is the difference in total weight between constraints where $x_i$ occurs positively and negatively, and $\bias(\Psi) \eqdef \frac1{km} \sum_{i=1}^n |\diff_i(\Psi)| \in [0,1]$.\footnote{\cite{GVV17,CGV20} did not normalize by $\frac1{kW}$.} In the sketching setting, $\bias(\Psi)$ can be estimated using standard $\ell_1$-norm sketching algorithms \cite{Ind06,KNW10}. 

In \cref{sec:thresh-alg}, we give simple optimal bias-based approximation algorithms for threshold functions:

\begin{theorem}\label{thm:thresh-bias-alg}
Let $f_{S,k} = \Th^i_k$ be a threshold function. Then for every $\epsilon > 0$, there exists a piecewise linear function $\gamma : [-1,1]\to[0,1]$ and a constant $\epsilon'>0$ such that the following is a sketching $(\alpha(f_{S,k})-\epsilon)$-approximation for $\mcsp(f_{S,k})$: On input $\Psi$, compute an estimate $\hat{b}$ for $\bias(\Psi)$ up to a multiplicative $(1\pm \epsilon')$ error and output $\gamma(\hat{b})$.
\end{theorem}

Our construction generalizes the algorithm in \cite{CGV20} for $\twoand$ to all threshold functions, and is also a simplification, since the \cite{CGV20} algorithm computes a more complicated function of $\hat{b}$.

For all CSPs whose approximability we resolve in this paper, we apply an analytical technique which we term the ``max-min method;'' see the discussion in \cref{sec:max-min} below. For such CSPs, our algorithm can be extended to solve the problem of outputting an approximately optimal \emph{assignment} (instead of just the value of such an assignment). Indeed, for this problem, we give a simple randomized streaming algorithm using $O(n)$ space and time:

\begin{theorem}[Informal version of \cref{thm:thresh-bias-output-alg}]\label{thm:thresh-bias-alg-classical}
Let $f_{S,k}$ be a function for which the max-min method applies, such as $\kand$, or $\Th^{k-1}_k$ (for even $k$). Then there exists a constant $p^* \in [0,1]$ such that following algorithm, on input $\Psi$, outputs an assignment with expected value at least $\alpha(f_{S,k}) \val_\Psi$: Assign variable $i$ to $1$ if $\diff_i(\Psi) \geq 0$ and $-1$ otherwise, and then flip each variable's assignment independently with probability $p^*$.
\end{theorem}

Our algorithm can potentially be derandomized using universal hash families, as in Biswas and Raman's recent derandomization \cite{BR21} of the $\maxtwoand$ algorithm in \cite{CGV20}.

\subsubsection{Sketching vs. streaming approximability}

\cref{thm:kand-approximability} implies that $\alpha(\threeand) = \frac29$. We prove that the padded one-wise pair criterion of Chou, Golovnev, Sudan, and Velusamy~\cite{CGSV21-boolean} is not sufficient to completely resolve the \emph{streaming} approximability of \m[$3$AND]:

\begin{theorem}[Informal version of \cref{thm:cgsv-streaming-failure-3and} + \cref{obs:cgsv-streaming-3and-lb}]
The padded one-wise pair criterion in \cite{CGSV21-boolean} does not rule out a $o(\sqrt n)$-space streaming $(\frac29+\epsilon)$-approximation for $\threeand$ for every $\epsilon > 0$; however, it does rule out such an algorithm for $\epsilon \gtrapprox 0.0141$.
\end{theorem}

We state these results formally in \cref{sec:cgsv-streaming-framework} and prove them in \cref{sec:cgsv-streaming-failure-3and}. Separately, \cref{thm:k-1-k-approximability} implies that $\alpha(\Th^3_4) = \frac49$, and the padded one-wise pair criterion \emph{can} be used to show that $(\frac49+\epsilon)$-approximating $\mcsp(\Th^3_4)$ requires $\Omega(\sqrt n)$ space in the streaming setting (see \cref{obs:th34-streaming-lb} below).

\subsection{Related work}

The classical approximability of ${\m}\kand$ has been the subject of intense study, both in terms of algorithms \cite{GW95,FG95,Zwi98,Tre98-alg,TSSW00,Has04,Has05,CMM09} and hardness-of-approximation \cite{Has01,Tre98-hardness,ST98,ST00,EH08,ST09}, given its intimate connections to $k$-bit PCPs. Charikar, Makarychev, and Makarychev~\cite{CMM09} constructed an $\Omega(k 2^{-k})$-approximation to ${\m}\kand$, while  Samorodnitsky and Trevisan~\cite{ST09} showed that $k2^{-(k-1)}$-approximations and $(k+1)2^{-k}$-approximations are $\NP$- and UG-hard, respectively.

Interestingly, recalling that $\alpha(\kand) \to 2\rho(\kand) = 2^{-(k-1)}$ as $k \to \infty$, in the large-$k$ limit our simple randomized algorithm (given in \cref{thm:thresh-bias-alg-classical}) matches the performance of Trevisan's~\cite{Tre98-alg} parallelizable LP-based algorithm for $\kand$, which (to the best of our knowledge) was the first work on the general $\kand$ problem! The subsequent works \cite{Has04,Has05,CMM09} superseding \cite{Tre98-alg} use more complex techniques involving semidefinite programming, but are structurally similar to our algorithm in \cref{thm:thresh-bias-alg-classical}: They all involve ``guessing'' an assignment $\vecx \in \mathbb{Z}_2^n$ and then perturbing each bit with constant probability.

\section{Our techniques}\label{sec:techniques}

In this section, we give a more detailed background on the technical aspects of the dichotomy theorem in \cite{CGSV21-boolean}, and explain the novel aspects of our analysis.

\subsection{The Chou, Golovnev, Sudan, and Velusamy~\cite{CGSV21-boolean} framework for symmetric functions}\label{sec:cgsv-framework}

In this section, we describe the Chou, Golovnev, Sudan, and Velusamy~\cite{CGSV21-boolean} framework for finding the optimal sketching approximation ratio of a symmetric Boolean function $f_{S,k}$.

Let $\Delta(\{-1,1\}^k)$ denote the space of all distributions on $\{-1,1\}^k$. For a distribution $\cD\in \Delta(\{-1,1\}^k)$ and $\vecx\in \{-1,1\}^k$, we use $\cD(\vecx)$ to denote the probability of sampling $\vecx$ in $\cD$. To a distribution $\cD \in \Delta(\{-1,1\}^k)$ we associate a \emph{canonical instance} $\Psi_{\cD}$ of $\mcsp(f_{S,k})$ on $k$ variables as follows. Let $\vecj = (1,\ldots,k)$. For every negation pattern $\vecb \in \{-1,1\}^k$, $\Psi_\cD$ contains the constraint $(\vecb,\vecj)$ with weight $\cD(\vecb)$.

We say a distribution $\cD \in \Delta(\{-1,1\}^k)$ is \emph{symmetric} if all vectors of equal Hamming weight are equiprobable, i.e., for every $\vecx, \vecy \in \{-1,1\}^k$ such that $\wt(\vecx)=\wt(\vecy)$, $\cD(\vecx)=\cD(\vecy)$. Let $\Delta_k \subseteq \Delta(\{-1,1\}^k)$ denote the set of all symmetric distributions on $\{-1,1\}^k$. Given $\cD \in \Delta_k$, let $\cD\langle i \rangle \eqdef \sum_{\vecx \in \{-1,1\}^k : \wt(\vecx) = i} \cD(\vecx)$ denote the total probability mass on vectors of Hamming weight $i$. Note that any vector $(\cD\langle 0 \rangle,\ldots,\cD\langle k \rangle)$ of nonnegative values summing to $1$ uniquely determines a distribution $\cD \in \Delta_k$; we write $\cD = (\cD\langle 0 \rangle,\ldots,\cD\langle k \rangle)$ for notational convenience.

Let $\bern(p)$ represent a random variable which is $1$ with probability $p$ and $-1$ with probability $1-p$. For $\cD \in \Delta(\{-1,1\})^k$ and $p \in [0,1]$, let
\begin{equation}\label{eqn:lambda-def}
    \lambda_S(\cD,p) \eqdef \Exp_{\veca \sim \cD, \vecb \sim  \bern(p)^k}[f_{S,k}(\veca \odot \vecb)] = \Exp_{\vecb \sim \bern(p)^k}[\val_{\Psi_\cD}(\vecb)]
\end{equation}
denote the expected value of a ``$p$-biased symmetric assignment'' on $\cD$'s canonical instance.
Also, for a symmetric distribution $\cD \in \Delta_k$, we define its (scalar) \emph{marginal}
\begin{equation}\label{eqn:mu-def}
\mu(\cD) \eqdef \Exp_{\vecb \sim \cD}[b_1] = \cdots = \Exp_{\vecb \sim \cD}[b_k].
\end{equation}

In general, $\lambda_S$ is linear in $\cD$ and degree-$k$ in $p$, and $\mu$ is linear in $\cD$. For $\cD \in \Delta_k$, we provide explicit formulas for $\lambda_S$ and $\mu$ in \cref{sec:formulas}.

Roughly, \cite{CGSV21-boolean} states that $\mcsp(f_{S,k})$ is hard to approximate in the sketching setting if there exist distributions $\cD_N,\cD_Y \in \Delta_k$ such that (1) $\mu(\cD_N) = \mu(\cD_Y)$ and (2) $\cD_Y$'s canonical instance is highly satisfied by the trivial (all-ones) assignment but (3) $\cD_N$'s canonical instance is not well-satisfied by any ``biased symmetric assignment''. To be precise, for $\cD \in \Delta(\{-1,1\}^k)$, let
\begin{equation}\label{eqn:beta_S-gamma_S-def}
    \beta_S(\cD) \eqdef \sup_{p \in [0,1]} \lambda_S (\cD,p) \text{ and } \gamma_S(\cD) \eqdef \lambda_S(\cD,1),
\end{equation} and define 
\begin{equation}\label{eqn:alpha-cdn-cdy}
    \alpha(f_{S,k}) \eqdef \inf_{\cD_N, \cD_Y \in \Delta_k:~ \mu(\cD_N) = \mu(\cD_Y)} \left(\frac{\beta_S(\cD_N)}{\gamma_S(\cD_Y)}\right).
\end{equation}

For every symmetric function $f_{S,k}$, \cite{CGSV21-boolean} proves that $ \alpha(f_{S,k})$ is the optimal sketching approximation ratio for $\mcsp(f_{S,k})$:

\begin{theorem}[{Combines \cite[Theorem 2.10 and Lemma 2.14]{CGSV21-boolean}}]\label{thm:alpha-optimize-over-dy-dn}
Let $f_{S,k} : \{-1,1\}^k\to\{0,1\}$ be a symmetric function. Then for every $\epsilon > 0$, there is an linear sketching $(\alpha(f_{S,k})-\epsilon)$-approximation to $\mcsp(f_{S,k})$ in $O(\log n)$ space, but any sketching $(\alpha(f_{S,k})+\epsilon)$-approximation to $\mcsp(f_{S,k})$ requires $\Omega(\sqrt n)$ space.
\end{theorem}

\begin{remark}
In the general case where $f:\{-1,1\}^k \to \{0,1\}$ is not symmetric, the approximability of $f$ is no longer characterized by \cref{eqn:alpha-cdn-cdy}. Instead, \cite{CGSV21-boolean} requires taking an infimum over \emph{all} (not necessarily symmetric) distributions $\cD_N,\cD_Y \in \Delta(\{-1,1\})^k$. Moreover, a general distribution $\cD \in \Delta(\{-1,1\})^k$ no longer has a single scalar marginal (as in \cref{eqn:mu-def}). Instead, we must consider a \emph{vector} marginal $\vecmu(\cD) = (\mu_1,\ldots,\mu_k)$ with $i$-th component $\mu_i = \Exp_{\vecb \sim \cD}[b_i]$; correspondingly, $\cD_N$ and $\cD_Y$ are required to satisfy the constraint $\vecmu(\cD_N) = \vecmu(\cD_Y)$. These issues motivate our focus on \emph{symmetric} functions in this paper. Since we need to consider only symmetric distributions in \cref{eqn:alpha-cdn-cdy}, $\cD_Y$ and $\cD_N$ are each parameterized by $k+1$ variables (as opposed to $2^k$ variables), and there is a single linear equality constraint (as opposed to $k$ constraints).
\end{remark}

\subsection{Formulations of the optimization problem}

In order to show that $\alpha(\twoand) = \frac49$, Chou, Golovnev, Sudan, and Velusamy~\cite[Example 1]{CGSV21-boolean} use the following reformulation of the optimization problem on the right hand side of \cref{eqn:alpha-cdn-cdy}. For a symmetric function $f_{S,k}$ and $\mu \in [-1,1]$, let
\begin{equation}\label{eqn:beta_Sk-gamma_Sk-def}
    \beta_{S,k}(\mu) = \inf_{\cD_N \in \Delta_k:~\mu(\cD_N) = \mu} \beta_S(\cD_N) \text{ and } \gamma_{S,k}(\mu) = \sup_{\cD_Y \in \Delta_k:~\mu(\cD_Y) = \mu} \gamma_S(\cD_Y);
\end{equation}
then
\begin{equation}\label{eqn:alpha-optimize-over-mu}
    \alpha(f_{S,k}) = \inf_{\mu \in [-1,1]} \left(\frac{\beta_{S,k}(\mu)}{\gamma_{S,k}(\mu)}\right).
\end{equation}

The optimization problem on the right-hand side of \cref{eqn:alpha-optimize-over-mu} appears simpler than that of \cref{eqn:alpha-cdn-cdy} because it is univariate, but there is a hidden difficulty: Finding an explicit solution requires giving explicit formulas for $\beta_{S,k}(\mu)$ and $\gamma_{S,k}(\mu)$. In the case of $\twoand = f_{\{2\},2}$, Chou, Golovnev, Sudan, and Velusamy~\cite{CGSV21-boolean} show that $\gamma_{\{2\},2}(\mu)$ is an explicit linear function of $\mu$; maximize the quadratic $\lambda_{\{2\}}(\cD_N,p)$ over $p \in [0,1]$ to find $\beta_{\{2\}}(\cD_N)$; and then minimize $\beta_{\{2\}}(\cD_N)$ given $\mu(\cD_N) = \mu$ to find $\beta_{\{2\},2}(\mu)$. However, while for general symmetric functions $f_{S,k}$ we can describe $\gamma_{S,k}(\mu)$ as an explicit piecewise linear function of $\mu$ (see \cref{lemma:gamma-formula} below), we do not know how to find closed forms for $\beta_{S,k}(\mu)$ even for $\threeand$. Thus, in this work we introduce a different formulation of the optimization problem:

\begin{equation}\label{eqn:alpha-optimize-over-dn}
    \alpha(f_{S,k}) =  \inf_{\cD_N \in \Delta_k} \left(\frac{\beta_S(\cD_N)}{\gamma_{S,k}(\mu(\cD_N))}\right).
\end{equation}

This reformulation is valid because
\begin{equation*}
    \alpha(f_{S,k}) = \inf_{\mu \in [-1,1],\cD_N \in \Delta_k:~\mu(\cD_N) = \mu} \left(\frac{ \beta_S(\cD_N)}{\gamma_{S,k}(\mu)}\right) = \inf_{\cD_N \in \Delta_k} \left(\frac{\beta_S(\cD_N)}{\gamma_{S,k}(\mu(\cD_N))}\right).
\end{equation*}

We view optimizing directly over $\cD_N \in \Delta_k$ as an important conceptual switch. In particular, our formulation emphasizes the calculation of $\beta_{S}(\cD_N)$ as the centrally difficult feature, yet we can still take advantage of the relative simplicity of calculating $\gamma_{S,k}(\mu)$.

\subsection{Our contribution: The max-min method}\label{sec:max-min}

\emph{A priori}, solving the optimization problem on the right-hand side of \cref{eqn:alpha-optimize-over-dn} still requires calculating $\beta_S(\cD_N)$, which involves maximizing a degree-$k$ polynomial. To get around this difficulty, we have made a key discovery, which was not noticed by Chou, Golovnev, Sudan, and Velusamy~\cite{CGSV21-boolean} even in the $\twoand$ case. Let $\cD_N^*$ minimize the right-hand side of \cref{eqn:alpha-optimize-over-dn}, and $p^*$ maximize $\lambda_S(\cD_N^*,\cdot)$. After substituting $ \beta_S(\cD) = \sup_{p \in [0,1]} \lambda_S (\cD,p)$ in \cref{eqn:alpha-optimize-over-dn}, and applying the max-min inequality, we get
\begin{equation}
\begin{aligned}
    \alpha(f_{S,k})  = \inf_{\cD_N \in \Delta_k}\sup_{p\in[0,1]} \left(\frac{\lambda_S(\cD_N,p)}{\gamma_{S,k}(\mu(\cD_N))}\right)
    &\geq \sup_{p\in[0,1]}  \inf_{\cD_N \in \Delta_k} \left(\frac{\lambda_S(\cD_N,p)}{\gamma_{S,k}(\mu(\cD_N))}\right)
    \\
    & \ge \inf_{\cD_N \in \Delta_k} \left(\frac{\lambda_S(\cD_N,p^*)}{\gamma_{S,k}(\mu(\cD_N))}\right)\, .\label{eqn:max-min}
\end{aligned}
\end{equation}

Given $p^*$, the right-hand side of \cref{eqn:max-min} is relatively easy to calculate, being a ratio of a linear and piecewise linear function of $\cD_N$. Our discovery is that, in a wide variety of cases, the quantity on the right-hand side of \cref{eqn:max-min} \emph{equals} $\alpha(f_{S,k})$; that is, $(\cD_N^*,p^*)$ is a \emph{saddle point} of $\frac{\lambda_S(\cD_N,p)}{\gamma_{S,k}(\mu(\cD_N))}$.\footnote{This term comes from the optimization literature; such points are also said to satisfy the ``strong max-min property'' (see, e.g., \cite[pp. 115, 238]{BV04}). The saddle-point property is guaranteed by von Neumann's minimax theorem for functions which are concave and convex in the first and second arguments, respectively, but this theorem and the generalizations we are aware of do not apply even to $\threeand$.}

This yields a novel technique, which we call the ``max-min method'', for finding a closed form for $\alpha(f_{S,k})$. First, we guess $\cD_N^*$ and $p^*$, and then, we show analytically that $\frac{\lambda_S(\cD_N,p)}{\gamma_{S,k}(\mu(\cD_N))}$ has a saddle point at $(\cD_N^*,p^*)$ and that $\lambda_S(\cD_N,p)$ is maximized at $p^*$. These imply that $\frac{\lambda_S(\cD_N^*,p^*)}{\gamma_{S,k}(\mu(\cD_N^*))}$ is a lower and upper bound on $\alpha(f_{S,k})$, respectively. For instance, in \cref{sec:kand-analysis}, in order to give a closed form for $\alpha(\kand)$ for odd $k$ (i.e., the odd case of \cref{thm:kand-approximability}), we guess $\cD_N^*\langle \frac{k+1}2 \rangle=1$ and $p^* = \frac{k+1}{2k}$ (by using Mathematica for small cases), and then check the saddle-point and maximization conditions in two separate lemmas (\cref{lemma:kand-lb,lemma:kand-ub}, respectively). Then, we show that $\alpha(\kand) = \alpha'_k$ by analyzing the right hand side of the appropriate instantiation of \cref{eqn:max-min}. We use similar techniques for $\kand$ for even $k$ (also \cref{thm:kand-approximability}) and for various other cases in \cref{sec:other-analysis,sec:k-1-k-analysis,sec:k+1/2-analysis}.

In all of these cases, the $\cD_N^*$ we construct is supported on at most two distinct Hamming weights, which is the property which makes finding $\cD_N^*$ tractable (using computer assistance). However, this technique is not a ``silver bullet'': it is not the case that the sketching approximability of every symmetric Boolean CSP can be exactly calculated by finding the optimal $\cD_N^*$ supported on two elements and using the max-min method. Indeed, (as mentioned in \cref{sec:other-analysis}) we verify using computer assistance that this is not the case for $f_{\{3\},4}$.

Finally, we remark that the saddle-point property is precisely what defines the value $p^*$ required for our simple classical algorithm for outputting approximately optimal assignments for $\mcsp(f_{S,k})$ where $f_{S,k} = \Th^i_k$ is a threshold function (see \cref{thm:thresh-bias-output-alg}).

\subsection{Streaming lower bounds}\label{sec:cgsv-streaming-framework}

Chou, Golovnev, Sudan, and Velusamy~\cite{CGSV21-boolean} also define the following condition on pairs $(\cD_N,\cD_Y)$, stronger than $\mu(\cD_N)=\mu(\cD_Y)$, which implies hardness of $\gbmaxf$ for \emph{streaming} algorithms:

\begin{definition}[Padded one-wise pairs, {\cite[\S2.3]{CGSV21-boolean}} (symmetric case)]\label{def:padded-onewise}
	A pair of distributions $(\cD_Y,\cD_N) \in \Delta_k$ forms a \emph{padded one-wise pair} if there exists $\tau\in[0,1]$ and distributions $\cD_0,\cD_Y',\cD_N' \in \Delta_k$ such that (1) $\mu(\cD_Y')=\mu(\cD_N')=0$ and (2) $\cD_Y = \tau \cD_0 + (1-\tau)\cD'_Y$ and $\cD_N = \tau \cD_0 + (1-\tau)\cD'_N$.
\end{definition}

\begin{theorem}[Streaming lower bound for padded one-wise pairs, {\cite[Theorem 2.11]{CGSV21-boolean}} (symmetric case)]\label{thm:padded-onewise-streaming}
Let $(\cD_Y,\cD_N)$ be a padded one-wise pair. Then for every $\epsilon > 0$, $(\beta_S(\cD_Y)+\epsilon,\gamma_S(\cD_N)-\epsilon)\dmaxf$ requires $\Omega(\sqrt{n})$ space in the streaming setting.
\end{theorem}

We prove that \cref{thm:padded-onewise-streaming} fails to rule out streaming $(\frac29+\epsilon)$-approximations to $\maxthreeand$ in the following sense:

\begin{theorem}\label{thm:cgsv-streaming-failure-3and}
There is no infinite sequence $(\cD_Y^{(1)},\cD_N^{(1)}),(\cD_Y^{(2)},\cD_N^{(2)}),\ldots$ of padded one-wise pairs on $\Delta_3$ such that \[ \lim_{t \to \infty} \frac{\beta_{\{3\}}(\cD_N^{(t)})}{\gamma_{\{3\}}(\cD_Y^{(t)})} = \frac29. \]
\end{theorem}

\cref{thm:cgsv-streaming-failure-3and} is proven formally in \cref{sec:cgsv-streaming-failure-3and}; here is a proof outline:

\begin{proof}[Proof outline]
As discussed in \cref{sec:max-min}, since $k=3$ is odd, to prove \cref{thm:kand-approximability} we show, using the max-min method, that $\cD_N^* = (0,0,1,0)$ minimizes $\frac{\beta_{\{3\}}(\cdot)}{\gamma_{\{3\},3}(\mu(\cdot))}$. We can show that the corresponding $\gamma_{\{3\},3}$ value is achieved by $\cD_Y^* = (\frac13,0,0,\frac23)$. In particular, $(\cD_N^*,\cD_Y^*)$ are not a padded one-wise pair.

We can show that the minimizer of $\gamma_{\{3\}}$ for a particular $\mu$ is in general unique. Hence, it suffices to furthermore show that $\cD_N^*$ is the \emph{unique} minimizer of $\frac{\beta_{\{3\}}(\cdot)}{\gamma_{\{3\},3}(\mu(\cdot))}$. For this purpose, the max-min method is not sufficient because $\frac{\lambda_{\{3\}}(\cdot,p^*)}{\gamma_{\{3\},3}(\mu(\cdot))}$ is not uniquely minimized at $\cD_N^*$ (where we chose $p^* = \frac23$). Intuitively, this is because $p^*$ is not a good enough estimate for the maximizer of $\lambda_{\{3\}}(\cD_N,\cdot)$. To remedy this, we observe that $\lambda_{\{3\}}((1,0,0,0),\cdot),\lambda_{\{3\}}((0,1,0,0),\cdot)$, $\lambda_{\{3\}}((0,0,1,0),\cdot)$ and $\lambda_{\{3\}}((0,0,0,1),\cdot)$ are minimized at $0,\frac13,\frac23$, and $1$, respectively. Hence, we instead lower-bound $\lambda_{\{3\}}(\cD_N,\cdot)$ by evaluating at $\frac13 \cD_N\langle 1 \rangle + \frac23 \cD_N\langle 2 \rangle + \cD_N\langle 3 \rangle$, which does suffice to prove the uniqueness of $\cD_N^*$. The theorem then follows from continuity arguments.
\end{proof}

Yet we still can achieve decent bounds using padded one-wise pairs:

\begin{observation}\label{obs:cgsv-streaming-3and-lb}
The padded one-wise pair $\cD_N=(0,0.45,0.45,0.1),\cD_Y=(0.45,0,0,0.55)$ (discovered by numerical search) \emph{does} prove a streaming approximability upper bound of $\approx .2362$ for $\threeand$, which is still quite close to $\alpha(\threeand)=\frac29$.
\end{observation}

\section{Formulas for $\mu$, $\lambda_S$, and $\gamma_{S,k}$}\label{sec:formulas}

In this section, we give explicit formulas for the quantities $\mu(\cD)$, $\lambda_S(\cD,p)$, and $\gamma_{S,k}(\mu)$ (defined in \cref{eqn:mu-def,eqn:lambda-def,eqn:beta_Sk-gamma_Sk-def}, respectively) which will be used throughout the rest of the paper. For $i \in [k]$, let $\epsilon_{i,k} \eqdef -1+\frac{2i}k$.

\begin{lemma}\label{lemma:mu-formula}
For any $\cD \in \Delta_k$, \[ \mu(\cD) = \sum_{i=0}^k \epsilon_{i,k} \,\cD\langle i \rangle. \]
\end{lemma}

\begin{proof}[Proof of \cref{lemma:mu-formula}]
By definition (\cref{eqn:mu-def}), $\mu(\cD) = \Exp_{\vecb\sim\cD}[b_1]$. We use linearity of expectation; the contribution of weight-$i$ vectors to $\mu(\cD)$ is $\cD\langle i \rangle \cdot \frac1k (i \cdot 1 + (k-i) \cdot (-1)) = \epsilon_{i,k} \,\cD\langle i \rangle$.
\end{proof}

\begin{lemma}\label{lemma:lambda-formula}
For any $\cD \in \Delta_k$ and $p \in [0,1]$, we have \[ \lambda_S(\cD,p) = \sum_{s \in S} \sum_{i=0}^k \left(\sum_{j=\max\{0,s-(k-i)\}}^{\min\{i,s\}} {i \choose j} {k-i \choose s-j} q^{s+i-2j} p^{k-s-i+2j} \right) \cD\langle i \rangle \] where $q \eqdef 1-p$.
\end{lemma}

\begin{proof}
By linearity of expectation and symmetry, it suffices to fix $s$ and $i$ and calculate, given a fixed string $\veca = (a_1,\ldots,a_k)$ of Hamming weight $i$ and a random string $\vecb = (b_1,\ldots,b_k) \sim \bern(p)^k$, the probability of the event $\wt(\veca \odot \vecb) = s$.

Let $A = \supp(\veca) = \{t \in [k] : a_t = 1\}$ and similarly $B = \supp(\vecb)$. We have $|A| = i$ and \[ s = \wt(\veca \odot \vecb) = |A \cap B| + |([k] \setminus A) \cap ([k] \setminus B)|. \] Let $j = |A \cap B|$, and consider cases based on $j$.

Given fixed $j$, we must have $|A \cap B| = j$ and $|([k] \setminus A) \cap ([k] \setminus B)| = s-j$. Thus if $j$ satisfies $j \leq i, s-j \leq k-i,j\geq0,j\leq s$, we have $\binom{i}{j}$ choices for $A \cap B$ and $\binom{k-i}{s-j}$ choices for $([k]\setminus A) \cap ([k] \setminus B)$; together, these completely determine $B$. Moreover $\wt(\vecb) = |B| = |B \cap A| + |B \cap ([k] \setminus A)| = j + (k-i)-(s-j) = k - s - i + 2j$, yielding the desired formula.
\end{proof}

\begin{lemma}\label{lemma:gamma-formula}
Let $S \subseteq [k]$, and let $s$ be its smallest element and $t$ its largest element (they need not be distinct). Then for $\mu \in [-1,1]$, \[\gamma_{S,k}(\mu) = \begin{cases}\frac{1+\mu}{1+\epsilon_{s,k}} & \mu \in [-1,\epsilon_{s,k}) \\
1 & \mu \in [\epsilon_{s,k},\epsilon_{t,k}] \\ \frac{1-\mu}{1-\epsilon_{t,k}} & \mu \in (\epsilon_{t,k},1] \end{cases} \] (which also equals $\min\left\{\frac{1+\mu}{1+\epsilon_{s,k}}, 1, \frac{1-\mu}{1-\epsilon_{t,k}}\right\}$).
\end{lemma}

\begin{proof}
For $\mu \in [-1,1]$, in (\cref{eqn:beta_Sk-gamma_Sk-def}) we defined \[ \gamma_{S,k}(\mu) = \sup_{\cD_Y \in \Delta_k : \mu(\cD_Y) = \mu} \gamma_S(\cD_Y), \] where by \cref{eqn:beta_S-gamma_S-def,eqn:lambda-def}), $\gamma_S(\cD_Y) = \sum_{i \in S} \cD_Y\langle i \rangle$. For $\cD_Y \in \Delta_k$, let $\supp(\cD_Y) = \{i \in [k]:\cD_Y\langle i \rangle > 0\}$. We handle cases based on $\mu$.

\subparagraph*{Case 1: $\mu \in [-1,\epsilon_{s,k}]$.} Our strategy is to reduce to the case $\supp(\cD_Y) \subseteq \{0,s\}$ while preserving the marginal $\mu$ and (non-strictly) increasing the value of $\gamma_S$.

Consider the following operation on a distribution $\cD_Y \in \Delta_k$: For $u < v < w \in [k]$, increase $\cD_Y\langle u \rangle$ by $\cD_Y\langle v \rangle \,\frac{w-v}{w-u}$, increase $\cD_Y\langle w \rangle$ by $\cD_Y\langle v \rangle\, \frac{v-u}{w-u}$, and set $\cD_Y\langle v \rangle$ to zero. Note that this results in a new distribution with the same marginal, since \[ \cD_Y\langle v \rangle \frac{w-v}{w-u} \epsilon_{u,k} + \cD_Y\langle v \rangle \frac{v-u}{w-u} \epsilon_{w,k} = \cD_Y\langle v \rangle \,\epsilon_{v,k}. \] Given an initial distribution $\cD_Y$, we can apply this operation to zero out $\cD_Y\langle v \rangle$ for $v \in \{1,\ldots,s-1\}$ by redistributing to $\cD_Y\langle 0 \rangle$ and $\cD_Y\langle s \rangle$, preserving the marginal and only increasing the value of $\gamma_S$ (since $v \not\in S$ while $s \in S$). Similarly, we can redistribute $\cD_Y\langle v \rangle$ to $\cD_Y\langle t \rangle$ and $\cD_Y\langle k \rangle$ when $v \in \{t+1,\ldots,k-1\}$, and to $\cD_Y\langle s \rangle$ and $\cD_Y\langle t \rangle$ when $v \in \{s+1,\ldots,t-1\}$. Thus, we need only consider the case $\supp(\cD) \subseteq \{0,s,t,k\}$. We assume for simplicity that $0,s,t,k$ are distinct.

By definition of $\epsilon_{i,k}$ we have \[ \mu(\cD) = -\cD_Y\langle 0 \rangle + \cD_Y\langle s \rangle\left(-1+\frac{2s}k\right) + \cD_Y\langle t \rangle\left(-1+\frac{2t}k\right) + \cD_Y\langle k \rangle \leq -1 + \frac{2s}k \] (by assumption for this case). Substituting $\cD_Y\langle s \rangle = 1-\cD_Y\langle 0 \rangle-\cD_Y\langle t \rangle-\cD_Y\langle k \rangle$ and multiplying through by $\frac{k}2$, we have \[ k\cD_Y\langle k \rangle-s\cD_Y\langle 0 \rangle-s\cD_Y\langle t \rangle-s\cD_Y\langle k \rangle+t\cD_Y\langle t \rangle \leq 0; \] defining $\delta = \cD_Y\langle t \rangle (\frac{t}s-1)+\cD_Y\langle k \rangle(\frac{k}s-1)$, we can rearrange to get $\cD_Y\langle 0 \rangle \geq \delta$. Then given $\cD_Y$, we can zero out $\cD_Y\langle t \rangle$ and $\cD_Y\langle k \rangle$, decrease $\cD_Y\langle 0 \rangle$ by $\delta$, and correspondingly increase $\cD_Y\langle s \rangle$ by $\cD_Y\langle t \rangle+\cD_Y\langle k \rangle+\delta$. This preserves the marginal since \[ (\delta + \cD_Y\langle t \rangle + \cD_Y\langle k \rangle) \,\epsilon_{s,k} = -\delta + \cD_Y\langle t \rangle \,\epsilon_{t,k} + \cD_Y\langle k \rangle \] and can only increase $\gamma_S$.

Thus, it suffices to only consider the case $\supp(\cD_Y) \subseteq \{0,s\}$. This uniquely determines $\cD_Y$ (because $\mu$ is fixed); we have $\cD_Y\langle 0 \rangle = \frac{\epsilon_{s,k}-\mu}{\epsilon_{s,k}+1}$ and $\cD_Y\langle s \rangle = \frac{1+\mu}{\epsilon_{s,k}+1}$, yielding the desired value of $\gamma_S$.

\subparagraph*{Case 2: $\mu \in [\epsilon_{s,k},\epsilon_{t,k}]$.} We simply construct $\cD_Y$ with $\cD_Y\langle s \rangle = \frac{\epsilon_{t,k}-\mu}{\epsilon_{s,k}-\epsilon_{t,k}}$ and $\cD_Y\langle t \rangle = \frac{\mu-\epsilon_{s,k}}{\epsilon_{s,k}-\epsilon_{t,k}}$; we have $\mu(\cD_Y) = \mu$ and $\gamma_S(\cD_Y) = 1$.

\subparagraph*{Case 3: $\mu \in [\epsilon_{t,k},1]$.} Following the symmetric logic to Case 1, we consider $\cD_Y$ supported on $\{t,k\}$ and set $\cD_Y\langle t \rangle = \frac{1-\mu}{1-\epsilon_{t,k}}$ and $\cD_Y\langle k \rangle = \frac{\mu-\epsilon_{t,k}}{1-\epsilon_{t,k}}$, yielding $\mu(\cD_Y) = \mu$ and $\gamma_S(\cD_Y) = \cD_Y\langle t \rangle$.
\end{proof}

\section{Analysis of $\alpha(\kand)$}\label{sec:kand-analysis}
In this section, we prove \cref{thm:kand-approximability} (on the sketching approximability of $\maxkand$). Recall that in \cref{eqn:alpha'_k}, we defined \[ \alpha'_k = \left(\frac{(k-1)(k+1)}{4k^2}\right)^{(k-1)/2}. \] \cref{thm:kand-approximability} follows immediately from the following two lemmas:

\begin{lemma}\label{lemma:kand-ub}
For all odd $k \geq 3$, $\alpha(\kand) \leq \alpha'_k$. For all even $k \geq 2$, $\alpha(\kand) \leq 2\alpha'_{k+1}$.
\end{lemma}

\begin{lemma}\label{lemma:kand-lb}
For all odd $k \geq 3$, $\alpha(\kand) \geq \alpha'_k$. For all even $k \geq 2$, $\alpha(\kand) \geq 2\alpha'_{k+1}$.
\end{lemma}

To begin, we give explicit formulas for $\gamma_{\{k\},k}(\mu(\cD))$ and $\lambda_{\{k\}}(\cD,p)$. Note that the smallest element of $\{k\}$ is $k$, and $\epsilon_{k,k} = 1$. Thus, for $\cD \in \Delta_k$, we have by \cref{lemma:gamma-formula,lemma:mu-formula} that
\begin{equation}\label{eqn:kand-gamma}
    \gamma_{\{k\},k}(\mu(\cD)) = \frac{1+\sum_{i=0}^k (-1+\frac{2i}k)\,\cD\langle i \rangle}{2} = \sum_{i=0}^k \frac{i}k \,\cD\langle i \rangle.
\end{equation}
Similarly, we can apply \cref{lemma:lambda-formula} with $s = k$; for each $i \in \kz$, $\max\{0,s-(k-i)\}=\min\{i,k\} = i$, so we need only consider $j=i$, and then $\binom{i}{j} = \binom{k-i}{s-j} = 1$. Thus, for $q = 1-p$, we have
\begin{equation}\label{eqn:kand-lambda}
    \lambda_{\{k\}}(\cD,p)  = \sum_{i=0}^k q^{k-i}p^i\,\cD\langle i \rangle
\end{equation}

Now, we prove \cref{lemma:kand-ub} directly:

\begin{proof}[Proof of \cref{lemma:kand-ub}]
Consider the case where $k$ is odd. Define $\cD_N^*$ by $\cD_N^*\langle \frac{k+1}2 \rangle=1$ and let $p^* = \frac12 + \frac1{2k}$. Since \[ \alpha(\kand) \leq \frac{\beta_{\{k\}}(\cD_N^*)}{\gamma_{\{k\},k}(\mu(\cD_N^*))} \text{ and } \beta_{\{k\}}(\cD_N) = \sup_{p \in [0,1]} \lambda_{\{k\}}(\cD_N^*,p), \] by \cref{eqn:alpha-optimize-over-dn,eqn:beta_S-gamma_S-def}, respectively, it suffices to check that $p^*$ maximizes $\lambda_{\{k\}}(\cD_N^*,\cdot)$ and \[ \frac{\lambda_{\{k\}}(\cD_N^*,p^*)}{\gamma_{\{k\},k}(\mu(\cD_N^*))} = \alpha'_k. \] Indeed, by \cref{eqn:kand-lambda}, \[ \lambda_{\{k\}}(\cD_N^*,p) = (1-p)^{(k-1)/2} p^{(k+1)/2}. \] To show $p^*$ maximizes $\lambda_{\{k\}}(\cD_N^*,\cdot)$, we calculate its derivative: \[ \frac{d}{dp}\left[(1-p)^{(k-1)/2} p^{(k+1)/2}\right] = - (1-p)^{(k-3)/2}p^{(k-1)/2}\left(kp-\frac{k+1}2\right), \] which has zeros only at $0,1,$ and $p^*$. Thus, $\lambda_{\{k\}}(\cD_N^*,\cdot)$ has critical points only at $0,1,$ and $p^*$, and it is maximized at $p^*$ since it vanishes at $0$ and $1$. Finally, by \cref{eqn:kand-gamma,eqn:kand-lambda} and the definition of $\alpha'_k$, \[ \frac{\lambda_{\{k\}}(\cD_N^*,p^*)}{\gamma_{\{k\},k}(\mu(\cD_N^*))} = \frac{\left(\frac12-\frac1{2k}\right)^{(k-1)/2} \left(\frac12+\frac1{2k}\right)^{(k+1)/2}}{\frac12 \left(1+\frac1k\right)} = \alpha'_k, \] as desired.

Similarly, consider the case where $k$ is even; here, we define $\cD_N^*$ by $\cD_N^*\langle \frac{k}2 \rangle = \frac{\left(\frac{k}2+1\right)^2}{\left(\frac{k}2\right)^2+\left(\frac{k}2+1\right)^2}$ and $\cD_N^*\langle \frac{k}2+1 \rangle = \frac{\left(\frac{k}2\right)^2}{\left(\frac{k}2\right)^2+\left(\frac{k}2+1\right)^2}$, and set $p^* = \frac12+\frac1{2(k+1)}$. Using \cref{eqn:kand-lambda} to calculate the derivative of $\lambda_{\{k\}}(\cD_N^*,\cdot)$ yields
\begin{multline*}
    \frac{d}{dp}\left[\frac{\left(\frac{k}2+1\right)^2}{\left(\frac{k}2\right)^2+\left(\frac{k}2+1\right)^2} (1-p)^{k/2}p^{k/2} + \frac{\left(\frac{k}2\right)^2}{\left(\frac{k}2\right)^2+\left(\frac{k}2+1\right)^2}(1-p)^{k/2-1}p^{k/2+1}\right] \\
    = -\frac{k}{2+2k+2k^2}(1-p)^{k/2-2} p^{k/2-1} \left(\frac{k}2+1-2p\right)\left((k+1)p-\left(\frac{k}2+1\right)\right),
\end{multline*} so $\lambda_{\{k\}}(\cD_N^*,\cdot)$ has critical points at $0,1,\frac12+\frac{k}4$. and $p^*$; $p^*$ is the only critical point in the interval $[0,1]$ for which $\lambda_{\{k\}}(\cD_N^*,\cdot)$ is positive, and hence is its maximum. Finally, it can be verified algebraically using \cref{eqn:kand-gamma,eqn:kand-lambda} that $\frac{\lambda_{\{k\}}(\cD_N^*,p^*)}{\gamma_{\{k\},k}(\mu(\cD_N^*))} = 2\alpha'_{k+1}$, as desired.
\end{proof}

We prove \cref{lemma:kand-lb} using the max-min method. We rely on the following proposition which is a simple inequality for optimizing ratios of linear functions, which we prove in \cref{sec:misc-proofs}:

\begin{proposition}\label{prop:lin-opt}
Let $f:\R^n \to \R$ be defined by the equation $f(\vecx) = \frac{\veca \cdot \vecx}{\vecb\cdot \vecx}$ for some $\veca,\vecb \in \R_{\geq 0}^n$. For every $\vecy(1),\ldots,\vecy(r) \in \R_{\geq 0}^n$, and every $\vecx = \sum_{i=1}^r \alpha_i \vecy(i)$ with each $x_i \geq 0$, we have $ f(\vecx) \geq \min_i f(\vecy(i))$. In particular, taking $r = n$ and $\vecy(1),\ldots,\vecy(n)$ as the standard basis for $\R^n$, for every $\vecx \in \R_{\geq 0}^n$, we have $f(\vecx) \geq \min_i \frac{a_i}{b_i}$.
\end{proposition}

\begin{proof}[Proof of \cref{lemma:kand-lb}]
First, suppose $k \geq 3$ is odd. Set $p^* = \frac12+\frac1{2k} = \frac{k+1}{2k}$. We want to show that
\begin{align*}
    \alpha'_k &\leq \inf_{\cD_N \in \Delta_k} \frac{\lambda_{\{k\}}(\cD_N,p^*)}{\gamma_{\{k\},k}(\mu(\cD_N))} \tag{max-min inequality, i.e., \cref{eqn:max-min}} \\
    &= \inf_{\cD_N \in \Delta_k} \frac{\sum_{i=0}^k (1-p^*)^{k-i}(p^*)^i\,\cD_N\langle i \rangle}{\sum_{i=0}^k \frac{i}k \,\cD_N\langle i \rangle} \tag{\cref{eqn:kand-gamma,eqn:kand-lambda}}.
\end{align*}

By \cref{prop:lin-opt}, it suffices to check that \[\forall i \in \kz,\quad (1-p^*)^{k-i}(p^*)^i \geq \alpha'_k \cdot \frac{i}k. \] By definition of $\alpha'_k$, we have that $\alpha'_k  = (1-p^*)^{(k-1)/2} (p^*)^{(k-1)/2}$. Defining $r = \frac{p^*}{1-p^*} = \frac{k+1}{k-1}$ (so that $p^* = r(1-p^*)$), factoring out $(1-p^*)^k$, and simplifying, we can rewrite our desired inequality as
\begin{equation}\label{eqn:kand-odd-lb-goal}
    \forall i \in \kz, \quad \frac12 (k-1) r^{i-\frac{k-1}2} \geq i.
\end{equation}
When $i = \frac{k+1}2$ or $\frac{k-1}2$, we have equality in \cref{eqn:kand-odd-lb-goal}. We extend to the other values of $i$ by induction. Indeed, when $i \geq \frac{k+1}2$, then ``$i$ satisfies \cref{eqn:kand-odd-lb-goal}'' implies ``$i+1$ satisfies \cref{eqn:kand-odd-lb-goal}'' because $r i \geq i + 1$, and when $i \leq \frac{k-1}2$, then ``$i$ satisfies \cref{eqn:kand-odd-lb-goal}'' implies ``$i-1$ satisfies \cref{eqn:kand-odd-lb-goal}'' because $\frac1r i \geq i - 1$.

Similarly, in the case where $k \geq 2$ is even, we set $p^* = \frac12+\frac1{2(k+1)}$ and $r = \frac{p^*}{1-p^*} = \frac{k+2}{k}$. In this case, for $i \in \kz$ the following analogue of \cref{eqn:kand-odd-lb-goal} can be derived: \[ \forall i \in \kz, \quad \frac12 k r^{i-\frac{k}2} \geq i, \] and these inequalities follow from the same inductive argument.
\end{proof}

\section{Further analyses of $\alpha(f)$ for symmetric Boolean functions $f$}

\subsection{$\Th^{k-1}_k$ for even $k$}\label{sec:k-1-k-analysis}

In this subsection, we prove \cref{thm:k-1-k-approximability} (on the sketching approximability of $\Th^{k-1}_k$ for even $k \geq 2$). It is necessary and sufficient to prove the following two lemmas:

\begin{lemma}\label{lemma:k-1-k-ub}
For all even $k \geq 2$, $\alpha(\Th^{k-1}_k) \leq \frac{k}{2} \alpha_{k-1}'$.
\end{lemma}

\begin{lemma}\label{lemma:k-1-k-lb}
For all even $k \geq 2$, $\alpha(\Th^{k-1}_k) \geq \frac{k}{2} \alpha_{k-1}'$.
\end{lemma}

Firstly, we give explicit formulas for $\gamma_{\{k-1,k\},k}$ and $\lambda_{\{k-1,k\}}$. We have $\Th_k^{k-1} = f_{\{k-1,k\},k}$, and $\epsilon_{k-1,k} = -1+\frac{2(k-1)}{k} = 1-\frac2k$. Thus, \cref{lemma:gamma-formula,lemma:mu-formula} give
\begin{equation}
    \gamma_{\{k-1,k\},k}(\mu(\cD)) = \min \left\{\frac{1 + \sum_{i=0}^k (-1+\frac{2i}k) \,\cD\langle i \rangle}{2-\frac2k} ,1\right\} = \min\left\{\sum_{i=0}^k \frac{i}{k-1}\,\cD\langle i \rangle,1\right\}.\label{eqn:k-1-k-gamma}
\end{equation}
Next, we calculate $\lambda_{\{k-1,k\}}(\cD,p)$ with \cref{lemma:lambda-formula}. Let $q = 1-p$, and let us examine the coefficient on $\cD\langle i \rangle$. $s=k$ contributes $q^{k-i}p^k$. In the case $i \leq k-1$, $s=k-1$ contributes $(k-i)q^{k-i-1}p^{i+1}$ for $j = i$, and in the case $i \geq 1$, $s=k-1$ contributes $iq^{k-i+1}p^{i-1}$ for $j = i-1$. Thus, altogether we can write
\begin{equation}
    \lambda_{\{k-1,k\}} (\cD,p) = \sum_{i=0}^k q^{k-i-1} p^{i-1} \left((k-i)p^2+ pq + iq^2\right) \,\cD\langle i \rangle.\label{eqn:k-1-k-lambda}
\end{equation}

Now, we prove \cref{lemma:k-1-k-ub,lemma:k-1-k-lb}.

\begin{proof}[Proof of \cref{lemma:k-1-k-ub}]
As in the proof of \cref{lemma:kand-ub}, it suffices to construct $\cD_N^*$ and $p^*$ such that $p^*$ maximizes $\lambda_{\{k-1,k\}}(\cD_N^*,\cdot)$ and $\frac{\lambda_{\{k-1,k\}}(\cD_N^*,p^*)}{\gamma_{\{k-1,k\},k}(\mu(\cD_N^*))} = \frac{k}2 \alpha'_{k-1}$.

We again let $p^* = \frac12+\frac1{2(k-1)}$, but define $\cD_N^*$ by $\cD_N^*\langle \frac{k}2 \rangle=\frac{\left(\frac{k}2\right)^2}{\left(\frac{k}2\right)^2+\left(\frac{k}2-1\right)^2}$ and $\cD_N^*\langle \frac{k}2+1 \rangle=\frac{\left(\frac{k}2-1\right)^2}{\left(\frac{k}2\right)^2+\left(\frac{k}2-1\right)^2}$. By \cref{eqn:k-1-k-lambda}, the derivative of $\lambda_{\{k-1,k\}}(\cD_N^*,\cdot)$ is now \begin{multline*}
    \frac{d}{dp}\Bigg[\frac{\left(\frac{k}2\right)^2}{\left(\frac{k}2\right)^2+\left(\frac{k}2-1\right)^2} (1-p)^{k/2-1} p^{k/2-1} \left(\frac{k}2p^2+pq+\frac{k}2q^2\right) + \\ \frac{\left(\frac{k}2-1\right)^2}{\left(\frac{k}2\right)^2+\left(\frac{k}2-1\right)^2}(1-p)^{k/2-2} p^{k/2} \left(\left(\frac{k}2-1\right)p^2+pq+\left(\frac{k}2+1\right)q^2\right)\Bigg] \\
    = -\frac1{8(k^2-2k+2)} (1-p)^{k/2-3}p^{k/2-2}(-k+(2(k-1)p) \xi(p),
\end{multline*} where $\xi(p)$ is the cubic \[ \xi(p) = -8k(k-1)p^3 + 2(k^3+k^2+6k-12)p^2 - 2(k^3-4)p + k^2(k-2). \] Thus, $\lambda_{\{k-1,k\}}$'s critical points on the interval $[0,1]$ are $0,1,p^*$ and any roots of $\xi$ in this interval. We claim that $\xi$ has no additional roots in the interval $(0,1)$. This can be verified directly by calculating roots for $k = 2,4$, so assume WLOG $k \geq 6$.

Suppose $\xi(p) = 0$ for some $p \in (0,1)$, and let $x = \frac1p - 1 \in (0,\infty)$. Then $p = \frac1{1+x}$; plugging this in for $p$ and multiplying through by $(x+1)^3$ gives the new cubic
\begin{equation}\label{eqn:for-descartes}
    (k^3-2k^2)x^3+(k^3-6k^2+8)x^2+(k^3-4k^2+12k-8)x+(k^3-8k^2+20k-16) = 0
\end{equation}
whose coefficients are cubic in $k$. It can be verified by calculating the roots of each coefficient of $x$ in \cref{eqn:for-descartes} that all coefficients are positive for $k \geq 6$. Thus, \cref{eqn:for-descartes} cannot have roots for positive $x$, a contradiction. Hence $\lambda_{\{k-1,k\}}(\cD_N^*,\cdot)$ is maximized at $p^*$. Finally, it can be verified that $\frac{\lambda_{\{k-1,k\}}(\cD_N^*,p^*)}{\gamma_{\{k-1,k\},k}(\mu(\cD_N^*))} = \frac{k}2\alpha'_{k-1}$, as desired.
\end{proof}

\begin{proof}[Proof of \cref{lemma:k-1-k-lb}]
Define $p^* = \frac12+\frac1{2(k-1)}$. Following the proof of \cref{lemma:kand-lb} and using the lower bound $\gamma_{\{k-1,k\},k}(\mu(\cD_N)) \leq \sum_{i=0}^k \frac{i}{k-1} \,\cD_N\langle i \rangle$, it suffices to show that \[ \frac{k}2 \alpha'_{k-1} \leq \inf_{\cD_N \in \Delta_k} \frac{\sum_{i=0}^k (1-p^*)^{k-i-1} (p^*)^{i-1} ((k-i)(p^*)^2+p^*(1-p^*)+i(1-p^*)^2) \,\cD_N\langle i \rangle}{\sum_{i=0}^k \frac{i}{k-1}\,\cD_N\langle i \rangle} \] for which by \cref{prop:lin-opt}, it in turn suffices to prove that for each $i \in \kz$, \[ \frac{k}2 \alpha'_{k-1} \frac{i}{k-1} \leq (1-p^*)^{k-i-1} (p^*)^{i-1} ((k-i)(p^*)^2+p^*(1-p^*)+i(1-p^*)^2). \] We again observe that $\alpha'_{k-1} = (1-p^*)^{k/2-1} (p^*)^{k/2-1}$, define $r = \frac{p^*}{1-p^*} = \frac{k}{k-2}$, and factor out $(1-p^*)^{k-1}$, which simplifies our desired inequality to
\begin{equation}\label{eqn:k-1-k-lb-goal}
    \frac{1}{2} r^{i-\frac{k}2-1} \cdot \frac{k-2}{k-1} \left(i + r + (k-i) r^2\right) \geq i.
\end{equation} for each $i \in \kz$. Again, we assume $k \geq 6$ WLOG; the bases cases $i = \frac{k}2-1,\frac{k}2$ can be verified directly, and we proceed by induction. If \cref{eqn:k-1-k-lb-goal} holds for $i$, and we seek to prove it for $i+1$, it suffices to cross-multiply and instead prove the inequality \[ r(i+1+r+(k-(i+1))r^2)i \geq (i+1) (i+r+(k-i)r^2), \] which simplifies to \[ (k-2i)(k-1)(k^2-4i-4) \leq 0, \] which holds whenever $ \frac{k}2\leq i \leq \frac{k^2-4}4$ (and $\frac{k^2-4}4 \geq k$ for all $k \geq 6$). The other direction (where $i \leq \frac{k}2-1$ and we induct downwards) is similar.
\end{proof}

\begin{observation}\label{obs:th34-streaming-lb}
For $\Th^3_4$ the optimal $\cD_N^* = (0,0,\frac45,\frac15,0)$ does participate in a padded one-wise pair with $\cD_Y^* = (\frac4{15},0,0,\frac{11}{15},0)$ (given by $\cD_0 = (0,0,0,1,0)$, $\tau = \frac15$, $\cD'_N = (0,0,1,0,0)$, and $\cD'_Y=(\frac4{15},0,0,\frac{8}{15},0)$) so we can rule out \emph{streaming} $(\frac49+\epsilon)$-approximations to $\mcsp(\Th^3_4)$ in $o(\sqrt n)$ space.
\end{observation}

\subsection{$\Ex^{(k+1)/2}_k$ for (small) odd $k$}\label{sec:k+1/2-analysis}

In this section, we prove bounds on the sketching approximability of $\Ex^{(k+1)/2}_k$ for odd $k\in \{3,\dots,51\}$. Define $\cD_{0,k} \in \Delta_k$ by $\cD_{0,k}\langle 0 \rangle=\frac{k-1}{2k}$ and $\cD_{0,k}\langle k \rangle = \frac{k+1}{2k}$. We prove the following two lemmas:

\begin{lemma}\label{lemma:k+1/2-ub}
For all odd $k \geq 3$, $\alpha(\Ex^{(k+1)/2}_k) \leq \lambda_{\{\frac{k+1}2\}}(\cD_{0,k},p'_k)$, where $p'_k \eqdef \frac{3 k - k^2 + \sqrt{4 k + k^2 - 2 k^3 + k^4}}{4k}$.
\end{lemma}

\begin{lemma}\label{lemma:k+1/2-lb}
The following holds for all odd $k \in \{3,\ldots,51\}$. For all $p \in [0,1]$, the expression $ \frac{\lambda_{\{\frac{k+1}2\}}(\cdot,p)}{\gamma_{\{\frac{k+1}2\},k}(\mu(\cdot))}$ is minimized at $\cD_{0,k}$.
\end{lemma}

We begin by writing an explicit formula for $\lambda_{\{\frac{k+1}2\}}$. \cref{lemma:lambda-formula} gives \[ \lambda_{\{\frac{k+1}2\}}(\cD,p) = \sum_{i=0}^k \left(\sum_{j=\max\{0,i-\frac{k-1}2\}}^{\min\{i,\frac{k+1}2\}} \binom{i}j \binom{k}{\frac{k+1}2-j} (1-p)^{(k+1)/2+i-2j} p^{(k-1)/2-i+2j}\right)\,\cD\langle i \rangle. \] For $i \leq \frac{k-1}2$, the sum over $j$ goes from $0$ to $i$, and for $i \geq \frac{k+1}2$, it goes from $i-\frac{k-1}2$ to $\frac{k+1}2$. Thus, plugging in $\cD_{0,k}$, we get:
\begin{equation}
    \lambda_{\{\frac{k+1}2\}}(\cD_{0,k},p) = \binom{k}{\frac{k+1}2} \left(\frac{k-1}{2k}(1-p)^{(k+1)/2}p^{(k-1)/2} + \frac{k+1}{2k}(1-p)^{(k-1)/2}p^{(k+1)/2}\right).\label{eqn:k+1/2-d0_k}
\end{equation}

By \cref{lemma:gamma-formula,lemma:mu-formula}, $\gamma_{\{\frac{k+1}2\},k}(\mu(\cD_{0,k})) = \gamma_{\{\frac{k+1}2\},k}(\frac1k) = 1$. Thus, \cref{lemma:k+1/2-ub,lemma:k+1/2-lb} together imply the following theorem:

\begin{theorem}\label{thm:k+1/2-approximability}
For odd $k \in \{3,\ldots,51\}$, \[ \alpha(\Ex^{(k+1)/2}_k) = \binom{k}{\frac{k+1}2} \left(\frac{k-1}{2k}(1-p'_k)^{(k+1)/2}(p'_k)^{(k-1)/2} + \frac{k+1}{2k}(1-p'_k)^{(k-1)/2}(p'_k)^{(k+1)/2}\right), \] where $p'_k = \frac{3 k - k^2 + \sqrt{4 k + k^2 - 2 k^3 + k^4}}{4k}$ as in \cref{lemma:k+1/2-ub}.
\end{theorem}

Recall that $\rho(f_{(k+1)/2,k}) = \binom{k}{\frac{k+1}2} 2^{-k}$. Although we currently lack a lower bound on $\alpha(\Ex^{(k+1)/2}_k)$ for large odd $k$, the upper bound from \cref{lemma:k+1/2-ub} suffices to prove \cref{thm:k+1/2-asymptotic-lb}, i.e., it can be verified that \[ \lim_{k \text{ odd} \to \infty} \frac{\binom{k}{\frac{k+1}2} \left( \frac{k-1}{2k}(1-p'_k)^{(k+1)/2}(p'_k)^{(k-1)/2} + \frac{k+1}{2k}(1-p'_k)^{(k-1)/2}(p'_k)^{(k+1)/2}\right)}{\rho(\Ex^{(k+1)/2}_k)} = 1. \]

We remark that for $\Ex^{(k+1)/2}_k$, our lower bound (\cref{lemma:k+1/2-lb}) is \emph{stronger} than what we were able to prove for $\kand$ (\cref{lemma:kand-lb}) and $\Th^{k-1}_k$ (\cref{lemma:k-1-k-lb}) because the inequality holds regardless of $p$. This is fortunate for us, as the optimal $p^*$ from \cref{lemma:k+1/2-ub} is rather messy.\footnote{The analogous statement is false for e.g. $\threeand$, where we had $\cD_N^* = (0,0,1,0)$, but at $p=\frac34$, \[ \frac{\lambda_{\{3\}}((0,\frac12,\frac12,0),\frac34)}{\gamma_{\{3\},3}(\mu(0,\frac12,\frac12,0))} = \frac3{16} \leq \frac{27}{128} =\frac{\lambda_{\{3\}}((0,0,1,0),\frac34)}{\gamma_{\{3\},3}(\mu(0,0,1,0))}. \]} It remains to prove \cref{lemma:k+1/2-ub,lemma:k+1/2-lb}.

\begin{proof}[Proof of \cref{lemma:k+1/2-ub}]
Taking the derivative with respect to $p$ of \cref{eqn:k+1/2-d0_k} yields \[ \frac{d}{dp}\left[\lambda_{\{\frac{k+1}2\}} (\cD_{0,k},p)\right] = - \frac1{4k} \binom{k}{\frac{k+1}2} (pq)^{(k-3)/2} (4kp^2 + (2k^2-6k)p + (-k^2+2k-1)), \] where $q=1-p$. Thus, $\lambda_{\{\frac{k+1}2\}} (\cD_{0,k},\cdot)$ has critical points at $p=0,1,p'_k,$ and $\frac{3 k - k^2 - \sqrt{4 k + k^2 - 2 k^3 + k^4}}{4k}$. This last value is nonpositive for all $k \geq 0$ (since $(3k-k^2)^2-(4 k + k^2 - 2 k^3 + k^4)=-4k(k-1)^2$).
\end{proof}

The proof of our lower bound (\cref{lemma:k+1/2-lb}) is slightly different than those of our earlier lower bounds (i.e., \cref{lemma:kand-lb,lemma:k-1-k-lb}) in the following sense. For $i \in \kz$, let $\cD_i \in \Delta_k$ be defined by $\cD_i\langle i \rangle=1$. For $\kand$ (\cref{lemma:kand-lb}), we used the fact that $\frac{\lambda_{\{k\}}(\cdot,p^*)}{\gamma_{\{k\},k}(\mu(\cdot))}$ is a ratio of linear functions, and thus using \cref{prop:lin-opt}, it is sufficient to verify the lower bound at $\cD_0,\ldots,\cD_k$. For $\Th_k^{k-1}$ (\cref{lemma:k-1-k-lb}), $\frac{\lambda_{\{k-1,k\}}(\cdot,p^*)}{\gamma_{\{k-1,k\},k}(\mu(\cdot))}$ is \emph{not} a ratio of linear functions, because the denominator $\gamma_{\{k-1,k\},k}(\mu(\cD)) = \min\{\sum_{i=0}^k \frac{i}{k-1} \,\cD\langle i \rangle, 1\}$ is not linear over $\Delta_k$. However, we managed to carry out the proof by upper-bounding the denominator with the linear function $\gamma'(\cD) = \sum_{i=0}^k \frac{i}{k-1} \cD\langle i \rangle$, and then invoking \cref{prop:lin-opt} (again, to show that it suffices to verify the lower bound at $\cD_0,\ldots,\cD_k$).

For $\Ex^{(k+1)/2}_k$, we show that it suffices to verify the lower bound on a larger (but still finite) set of distributions.

\begin{proof}[Proof of \cref{lemma:k+1/2-lb}]
Recalling that $\epsilon_{(k+1)/2,k} = \frac1k$, let $\Delta^+_k = \{\cD \in \Delta_k : \mu(\cD) \leq \frac1k\}$ and $\Delta^-_k = \{\cD \in \Delta_k : \mu(\cD) \geq \frac1k\}$. Note that $\Delta_k^+ \cup \Delta_k^- = \Delta_k$, and restricted to either $\Delta_k^+$ or $\Delta_k^-$, $\gamma_{\{\frac{k+1}2\},k}(\mu(\cdot))$ is linear and thus we can apply \cref{prop:lin-opt} to $\frac{\lambda_{\{k-1,k\}}(\cdot,p^*)}{\gamma_{\{k-1,k\},k}(\mu(\cdot))}$.

Let $\cD_{i,j} \in \Delta_k$, for $i < \frac{k+1}2, j > \frac{k+1}2$, be defined by $\cD_{i,j}\langle i \rangle=\frac{2j-(k+1)}{2(j-i)}$ and $\cD_{i,j}\langle j \rangle=\frac{(k+1)-2i}{2(j-i)}$. Note that $\mu(\cD_{i,j}) = \frac1k$ for each $i,j$. We claim that $\{\cD_i\}_{i \leq \frac{k+1}2} \cup \{\cD_{i,j}\}$ are the extreme points of $\Delta_k^+$, or more precisely, that every distribution $\cD \in \Delta_k^+$ can be represented as a convex combination of these distributions. Indeed, this follows constructively from the procedure which, given a distribution $\cD$, subtracts from each $\cD\langle i \rangle$ for $i < \frac{k+1}2$ (adding to the coefficient of the corresponding $\cD_i$) until the marginal of the (renormalized) distribution is $\frac1k$, and then subtracts from pairs $\cD\langle i \rangle,\cD\langle j \rangle$ with $i < \frac{k+1}2$ and $j > \frac{k+1}2$, adding it to the coefficient of the appropriate $\cD_{i,j}$) until $\cD$ vanishes (i.e., $\cD\langle i \rangle$ is zero for all $i \in \kz$). Similarly, every distribution $\cD \in \Delta_k^-$ can be represented as a convex combination of the distributions $\{\cD_i\}_{i \geq \frac{k+1}2} \cup \{\cD_{i,j}\}$. Thus, by \cref{prop:lin-opt}, it is sufficient to verify that \[\frac{\lambda_{\{\frac{k+1}2\}}(\cD,p)}{\gamma_{\{\frac{k+1}2\},k}(\mu(\cD))} \geq \frac{\lambda_{\{\frac{k+1}2\}}(\cD_N^*,p)}{\gamma_{\{\frac{k+1}2\},k}(\mu(\cD_N^*))} \] for each $\cD \in \{\cD_i\}\cup\{\cD_{i,j}\}$. Treating $p$ as a variable, for each odd $k\in\{3,\ldots,51\}$ we produce a list of $O(k^2)$ degree-$k$ polynomial inequalities in $p$ which we verify using Mathematica.
\end{proof}

\subsection{More symmetric functions}\label{sec:other-analysis}

In \cref{table:other-sym-funcs} below, we list four more symmetric Boolean functions (beyond $\kand$, $\Th^{k-1}_k$, and $\Ex^{(k+1)/2}_k$) whose sketching approximability we have analytically resolved using the ``max-min method''. These values were calculated using two functions in the Mathematica code, \texttt{estimateAlpha} --- which numerically or symbolically estimates the $\cD_N$, with a given support, which minimizes $\alpha$ --- and \texttt{testMinMax} --- which, given a particular $\cD_N$, calculates $p^*$ for that $\cD_N$ and checks analytically whether lower-bounding by evaluating $\lambda_S$ at $p^*$ proves that $\cD_N$ is minimal.

\begin{table}[h]
    \centering
    \begin{tabular}{|c|c|c|c|}
        \hline
         $S$ & $k$ & $\alpha$ & $\cD_N^*$ \\ \hline
         $\{2,3\}$ & $3$ & $\frac12+\frac{\sqrt3}{18} \approx 0.5962$ & $(0,\frac12,0,\frac12)$ \\ 
         $\{4,5\}$ & $5$ & $8\rroot(P_1) \approx 0.2831$ & $(0,0,1-\rroot(P_2),\rroot(P_2),0,0)$ \\
         $\{4\}$ & $5$ & $8\rroot(P_3) \approx 0.2394$ & $(0,0,1-\rroot(P_4),\rroot(P_4),0,0)$ \\
         $\{3,4,5\}$ & $5$ & $\frac12+\frac{3\sqrt5}{125} \approx 0.5537$ & $(0,\frac12,0,0,0,\frac12)$\\
         \hline
    \end{tabular}
    \caption{Symmetric functions for which we have analytically calculated exact $\alpha$ values using the ``max-min method''. For a polynomial $P : \R \to \R$ with a \emph{unique} positive real root, let $\rroot(p)$ denote that root, and define the polynomials $P_1(z)=-72+4890z-108999z^2+800000z^3, P_2(z)=-908+5021z-9001z^2+5158z^3$, $P_3(z) = -60+5745z-183426z^2+1953125z^3$, $P_4(z) = -344+1770z-3102z^2+1811z^3$. (We note that in the $f_{\{4\},5}$ and $f_{\{4,5\},5}$ calculations, we were required to check equality of roots numerically (to high precision) instead of analytically).}
    \label{table:other-sym-funcs}
\end{table}

We remark that two of the cases in \cref{table:other-sym-funcs} (as well as $\kand$), the optimal $\cD_N$ is rational and supported on two coordinates. However, in the other two cases in \cref{table:other-sym-funcs}, the optimal $\cD_N$ involves roots of a cubic.

In \cref{sec:k+1/2-analysis}, we showed that $\cD_N^*$ defined by $\cD_N^*\langle 0 \rangle=\frac{k-1}{2k}$ and $\cD_N^*\langle k \rangle=\frac{k+1}{2k}$ is optimal for $\Ex^{(k+1)/2}_k$ for odd $k \in \{3,\ldots,51\}$. Using the same $\cD_N^*$, we are also able to resolve 11 other cases in which $S$ is ``close to'' $\{\frac{k+1}2\}$; for instance, $S=\{5,6\},\{5,6,7\},\{5,7\}$ for $k=9$. (We have omitted the values of $\alpha$ and $\cD_N$ because they are defined using the roots of polynomials of degree up to 8.)

In all previously-mentioned cases, the condition ``$\cD_N^*$ has support size $2$'' was helpful, as it makes the optimization problem over $\cD_N^*$ essentially univariate; however, we have confirmed analytically in two other cases ($S=\{3\},k=4$ and $S=\{3,5\},k=5$) that ``max-min method on distributions with support size two'' does not suffice for tight bounds on $\alpha$ (see \texttt{testDistsWithSupportSize2} in the Mathematica code). However, using the max-min method with $\cD_N$ supported on two levels still achieves decent (but not tight) bounds on $\alpha$. For $S = \{3\},k=4$, using $\cD_N = (\frac14,0,0,0,\frac34)$, we get the bounds $\alpha(f_{\{3\},4}) \in [0.3209,0.3295]$ (the difference being $2.67\%$). For $S = \{3,5\},k=5$, using $\cD_N = (\frac14,0,0,0,\frac34,0)$, we get $\alpha(f_{\{3,5\},5}) \in [0.3416,0.3635]$ (the difference being $6.42\%$).

Finally, we have also analyzed cases where we get numerical solutions which are very close to tight, but we lack analytical solutions because they likely involve roots of high-degree polynomials. For instance, in the case $S = \{4,5,6\}, k=6$,  setting $\cD_N = (0,0,0,0.930013,0,0,0.069987)$  gives $\alpha(f_{\{4,5,6\},6}) \in [0.44409972,0.44409973]$, differing only by 0.000003\%. (We conjecture here that $\alpha=\frac49$.) For $S=\{6,7,8\},k=8$, using
$\cD_N=(0,0,0,0,0.699501,0.300499)$, we get the bounds $\alpha(f_{\{6,7,8\},8}) \in [0.20848,0.20854]$ (the difference being $0.02\%$).\footnote{Interestingly, in this latter case, we get bounds differing by $2.12\%$ using $\cD_N=(0,0,0,0,\frac9{13},\frac4{13},0,0,0)$ in an attempt to continue the pattern from $f_{\{7,8\},8}$ and $f_{\{8\},8}$ (where we set $\cD_N^* = (0,0,0,0,\frac{16}{25},\frac9{25},0,0,0)$ and $(0,0,0,0,\frac{25}{41},\frac{16}{41},0,0,0)$ in \cref{sec:k-1-k-analysis} and \cref{sec:kand-analysis}, respectively).}

\section{Incompleteness of streaming lower bounds: Proving \cref{thm:cgsv-streaming-failure-3and}}\label{sec:cgsv-streaming-failure-3and}

In this section, we prove \cref{thm:cgsv-streaming-failure-3and}, showing that the streaming lower bounds from \cite{CGSV21-boolean} (\cref{thm:padded-onewise-streaming}) cannot characterize the \emph{streaming} approximability of $\threeand$.

\begin{lemma}\label{lemma:3and-unique-minimum}
For $\cD \in \Delta_3$, the expression \[  \frac{\lambda_{\{3\}}(\cD,\frac13\cD\langle 1 \rangle+\frac23\cD\langle 2 \rangle+\cD\langle 3 \rangle)}{\gamma_{\{3\},3}(\mu(\cD))} \] is minimized uniquely at $\cD = (0,0,1,0)$, with value $\frac29$.
\end{lemma}

\begin{proof}
Letting $p = \frac13\cD\langle 1 \rangle+\frac23\cD\langle 2 \rangle+\cD\langle 3 \rangle$ and $q = 1-p$, by \cref{lemma:lambda-formula,lemma:gamma-formula,lemma:mu-formula} the expression expands to \[ \frac{\cD\langle 0 \rangle \, p^3 + \cD\langle 1 \rangle \, p^2(1-p)+\cD\langle 2 \rangle \, p(1-p)^2 + \cD\langle 3 \rangle \, (1-p)^3}{\frac12(1-\cD\langle 0 \rangle -\frac13\cD\langle 1 \rangle+\frac13\cD\langle 2 \rangle+\cD\langle 3 \rangle)}. \] The expression's minimum, and its uniqueness, are confirmed analytically in the Mathematica code.
\end{proof}

\begin{lemma}\label{lemma:3and-top-lemma}
Let $X$ be a compact topological space, $Y \subseteq X$ a closed subspace, $Z$ a topological space, and $f : X \to Z$ a continuous map. Let $x^* \in X, z^* \in Z$ be such that $f^{-1}(z^*) = \{x^*\}$. Let $\{x_i\}_{i \in \N}$ be a sequence of points in $Y$ such that $\{f(x_i)\}_{i \in \N}$ converges to $z^*$. Then $x^* \in Y$.
\end{lemma}

\begin{proof}
By compactness of $X$, there is a subsequence $\{x_{j_i}\}_{i \in \N}$ which converges to a limit $\tilde{x}$. By closure, $\tilde{x} \in Y$. By continuity, $f(\tilde x) = z^*$, so $\tilde x = x^*$.
\end{proof}

Finally, we have:

\begin{proof}[Proof of \cref{thm:cgsv-streaming-failure-3and}]
By \cref{lemma:3and-unique-minimum}, $\frac{\beta_{\{3\}}(\cD_N)}{\gamma_{\{3\},3}(\mu(\cD_N))}$ is minimized \emph{uniquely} at $\cD_N^* = (0,0,1,0)$. By \cref{lemma:mu-formula} we have $\mu(\cD_N^*) = \frac13$, and by inspection from the proof of \cref{lemma:gamma-formula} below, $\gamma_{\{3\}}(\cD_Y)$ with $\mu(\cD_Y)=\frac13$ is uniquely minimized by $\cD_Y^*=(\frac13,0,0,\frac23)$.

Finally, we rule out the possibility of an infinite sequence of padded one-wise pairs which achieve ratios arbitrarily close to $\frac29$ using topological properties. View a distribution $\cD \in \Delta_3$ as the vector $(\cD\langle 0 \rangle,\cD\langle 1 \rangle,\cD\langle 2 \rangle,\cD\langle 3 \rangle) \in \R^4$. Let $D \subset \R^4$ denote the set of such distributions. Let $M\subset D \times D \subset\R^8$ denote the subset of pairs of distributions with matching marginals, and let $M' \subset M$ denote the subset of pairs with uniform marginals and $P \subset M$ the subset of padded one-wise pairs. $D$, $M$, $M'$, and $P$ are compact (under the Euclidean topology); indeed, $D$, $M$, and $M'$ are bounded and defined by a finite collection of linear equalities and strict inequalities, and letting $M' \subset M$ denote the subset of pairs of distributions with matching \emph{uniform} marginals, $P$ is the image of the compact set $[0,1] \times D \times M' \subset \R^{13}$ under the continuous map $\tau \times \cD_0 \times (\cD'_Y,\cD'_N) \mapsto (\tau \cD_0 + (1-\tau) \cD'_Y,\tau \cD_0+(1-\tau)\cD'_N)$. Hence, $P$ is closed.

Now the function \[ \alpha : M \to \R \cup \{\infty\}: (\cD_N,\cD_Y) \mapsto \frac{\beta_{\{3\}}(\cD_N)}{\gamma_{\{3\}}(\cD_Y)} \] is continuous, since a ratio of continuous functions is continuous, and $\beta_{\{3\}}$ is a single-variable supremum of a continuous function (i.e., $\lambda_S$) over a compact interval, which is in general continuous in the remaining variables. Thus, if there were a sequence of padded one-wise pairs $\{(\cD_N^{(i)},\cD_Y^{(i)}) \in P\}_{i \in \N}$ such that $\alpha(\cD_N^{(i)},\cD_Y^{(i)})$ converges to $\frac29$ as $i \to \infty$, since $M$ is compact and $P$ is closed, \cref{lemma:3and-top-lemma,lemma:3and-unique-minimum} imply that $(\cD_N^*,\cD_Y^*) \in P$, a contradiction.
\end{proof}

\section{Simple sketching algorithms for threshold functions}\label{sec:thresh-alg}

The main goal of this section is to prove \cref{thm:thresh-bias-alg}, giving a simple ``bias-based'' sketching algorithm for threshold functions $\Th^i_k$. Given an instance $\Psi$ of $\mcsp(\Th^i_k)$, for $i \in [n]$, let $\diff_i(\Psi)$ denote the total weight of clauses in which $x_i$ appears positively minus the weight of those in which it appears negatively; that is, if $\Psi$ consists of clauses $(\vecb(1),\vecj(1)),\ldots,(\vecb(m),\vecj(m))$ with weights $w_1,\ldots,w_m$, then \[ \diff_i(\Psi) \eqdef \sum_{\ell\in[m] \text{ s.t. } j(\ell)_t = i \text{ for some } t\in[k]}  b(\ell)_t w_\ell.  \] Let $\bias(\Psi) \eqdef \frac1{kW} \sum_{i=1}^n |\diff_i(\Psi)|$, where $W = \sum_{\ell=1}^m w_\ell$ is the total weight in $\Psi$.

Let $S = \{i,\ldots,k\}$ so that $\Th^i_k = f_{S,k}$. Recall the definitions of $\beta_{S,k}(\mu)$ and $\gamma_{S,k}(\mu)$ from \cref{eqn:alpha-optimize-over-mu}. Our simple algorithm for $\mcsp(\Th^i_k)$ relies on the following two lemmas, which we prove below:

\begin{lemma}\label{lemma:thresh-value-ub}
$\val_\Psi \leq \gamma_{S,k}(\bias(\Psi))$.
\end{lemma}

\begin{lemma}\label{lemma:thresh-value-lb}
$\val_\Psi \geq \beta_{S,k}(\bias(\Psi))$.
\end{lemma}

Together, these two lemmas imply that outputting $\alpha(\Th^i_k)\cdot \gamma_{S,k}(\bias(\Psi))$ gives an $\alpha(\Th^i_k)$-approximation to $\mcsp(\Th^i_k)$, since $\alpha(\Th^i_k) = \inf_{\mu \in [-1,1]} \frac{\beta_{S,k}(\mu)}{\gamma_{S,k}(\mu)}$ (\cref{eqn:alpha-optimize-over-mu}). We can implement this as a small-space sketching algorithm (up to an arbitrarily small constant $\epsilon > 0$ in the approximation ratio) because $\bias(\Psi)$ is measurable using $\ell_1$-sketching algorithms (as used also in \cite{GVV17,CGV20,CGSV21-boolean}) and $\gamma_{S,k}(\cdot)$ is piecewise linear:

\begin{theorem}[{\cite{Ind06,KNW10}}]\label{thm:l1-sketching}
For every $\epsilon>0$, there exists an $O(\log n/\epsilon^2)$-space randomized sketching algorithm for the following problem: The input is a stream $S$ of updates of the form $(i, v) \in [n] \times \{-\mathrm{poly}(n),\ldots,\mathrm{poly}(n)\}$, and the goal is to estimate the $\ell_1$-norm of the vector $x \in [n]^n$ defined by $x_i = \sum_{(i,v) \in S} v$, up to a multiplicative factor of $1\pm\epsilon$.
\end{theorem}

\begin{corollary}\label{cor:bias-sketching}
For $f : \{-1,1\}^k \to \{0,1\}$ and every $\epsilon>0$, there exists an $O(\log n/\epsilon^2)$-space randomized sketching algorithm for the following problem: The input is an instance $\Psi$ of $\mcsp(\Th^i_k)$ (given as a stream of constraints), and the goal is to estimate $\bias(\Psi)$ up to a multiplicative factor of $1\pm\epsilon$.
\end{corollary}

\begin{proof}
Invoke the $\ell_1$-norm sketching algorithm from \cref{thm:l1-sketching} as follows: On each input constraint $(\vecb=(b_1,\ldots,b_k),\vecj=(j_1,\ldots,j_k))$ with weight $w$, insert the updates $(j_1,wb_1),\ldots,(j_k,wb_k)$ into the stream (and normalize appropriately).
\end{proof}

\cref{thm:thresh-bias-alg} then follows from \cref{lemma:thresh-value-lb,lemma:thresh-value-ub,cor:bias-sketching}; we include a formal proof in \cref{sec:misc-proofs} for completeness.

To prove \cref{lemma:thresh-value-lb,lemma:thresh-value-ub}, we require a bit more setup. Adapting notation from \cite[\S4.2]{CGSV21-boolean}, given an instance $\Psi$ of $\mcsp(\Th^i_k)$ and a ``negation pattern'' $\veca = (a_1,\ldots,a_n) \in \{-1,1\}^n$ for the variables, let $\Psi^\veca$ be the instance which results from $\Psi$ by ``flipping'' the variables according to $\veca$ (formally, each constraint $(\vecb,\vecj)$ is replaced with $(\vecb\odot\veca|_{\vecj},\vecj)$). We summarize the useful properties of this operation in the following claim:

\begin{proposition}\label{prop:flips}
Let $\Psi$ be an instance of $\mcsp(\Th^i_k)$ and $\veca = (a_1,\ldots,a_n) \in \{-1,1\}^n$. Then:
\begin{enumerate}[label=\roman*.,ref=\roman*]
    \item For each $i \in [n]$, $\diff_i(\Psi^\veca) = a_i \diff_i(\Psi)$.\label{item:flip-diff}
    \item $\bias(\Psi) = \bias(\Psi^\veca)$.\label{item:flip-bias}
    \item For any $\vecsigma \in \{-1,1\}^n$, $\val_{\Psi^\veca}(\vecsigma) = \val_{\Psi}(\veca\odot\vecsigma)$.\label{item:flip-val-ass}
    \item $\val_{\Psi^\veca} = \val_\Psi$.\label{item:flip-val}
\end{enumerate}
\end{proposition}

\begin{proof}
For \cref{item:flip-diff}, we have
\begin{align*}
    \diff_i(\Psi^\veca) &= \sum_{\ell\in[m] \text{ s.t. }j(\ell)_t = i\text{ for some }t\in[k]} a_{j(\ell)_t} b(\ell)_t w_\ell \tag{definition of $\diff_i$} \\
    &= a_i \; \sum_{\ell\in[m] \text{ s.t. }j(\ell)_t = i\text{ for some }t\in[k]} b(\ell)_t w_\ell \\
    &= a_i \diff_i(\Psi) \tag{definition of $\diff_i$}.
\end{align*}
\cref{item:flip-bias} follows immediately from \cref{item:flip-diff} and the definition $\bias(\Psi) = \frac1{kW} \sum_{i=1}^n |\diff_i(\Psi)|$. For \cref{item:flip-val-ass}, we have
\begin{align*}
    \val_{\Psi^\veca}(\vecsigma) &= \frac1W\sum_{i\in[m]} w_i \Th^i_k((\vecb(i) \odot \veca|_{\vecj(i)})\odot\vecsigma|_{\vecj(i)}) \tag{definitions of $\Psi^\veca$ and $\val$} \\
    &= \frac1W\sum_{i\in[m]} w_i \Th^i_k(\vecb(i)\odot(\vecsigma\odot\veca)|_{\vecj(i)}) \\
    &= \val_\Psi(\vecsigma \odot \veca).\tag{definition of $\val$}
\end{align*}
Finally, \cref{item:flip-val} follows from \cref{item:flip-val-ass} and the fact that $\{\vecsigma : \vecsigma \in \{-1,1\}^n\} = \{\vecsigma \odot \veca : \vecsigma \in \{-1,1\}^n\}$: \[ \val_{\Psi^\veca} = \max_{\vecsigma \in\{-1,1\}^n} \val_{\Psi^\veca}(\vecsigma) = \max_{\vecsigma \in\{-1,1\}^n} \val_{\Psi}(\vecsigma \odot \veca) = \max_{\vecsigma \in\{-1,1\}^n} \val_{\Psi}(\vecsigma) = \val_\Psi. \]
\end{proof}

Also, given an instance $\Psi$, we define its ``symmetrized canonical distribution'' $\dsym_\Psi \in \Delta_k$ to be the distribution obtained by sampling a constraint at random from $\Psi$ and outputting its ``randomly permuted negation pattern''. Formally, let $\sym_k$ denote the set of permutations $[k] \to [k]$. For a vector $\vecb =(b_1,\ldots,b_k)\in \{-1,1\}^k$ and a permutation $\vecpi \in \sym_k$, let $\vecpi(\vecb) = (b_{\vecpi(1)},\ldots,b_{\vecpi(k)})$. Let $C(i) = (\vecb(i),\vecj(i))$ denote the $i$-th constraint of $\Psi$, with weight $w_i$, and let $W = \sum_{i=1}^m w_i$ be the total weight. To sample a random vector from $\dsym_\Psi$, we sample $i \in [m]$ with probability $w_i/W$, sample a permutation $\vecpi \sim \Unif(\sym_k)$, and output $\vecpi(\vecb(i))$. The useful properties of $\dsym_\Psi$ are summarized in the following claim:

\begin{proposition}\label{prop:dsym}
Let $\Psi$ be an instance of $\mcsp(\Th^i_k)$. Then:
\begin{enumerate}[label=\roman*.,ref=\roman*]
\item For any $p \in [0,1]$, $\Exp_{\veca \sim \bern(p)^n}[\val_\Psi(\veca)] = \lambda_S(\dsym_\Psi,p)$.\label{item:dsym-val}
\item $\mu(\dsym_\Psi) = \frac1{kW} \sum_{i=1}^n \diff_i(\Psi) \leq \bias(\Psi)$.\label{item:dsym-mu}
\item If $\diff_i(\Psi) \geq 0$ for all $i \in [n]$, then $\mu(\dsym_\Psi) = \bias(\Psi)$.\label{item:dsym-bias}
\end{enumerate}
\end{proposition}

\begin{proof}
We begin with \cref{item:dsym-val}. Fix a constraint $(\vecb,\vecj)$. We make two observations. Firstly, if we sample $\veca\sim\bern(p)^n$, the distribution of $\veca|_\vecj$ is identical to $\bern(p)^k$. Secondly, for a fixed vector $\vecb \in \{-1,1\}^k$, if we sample $\veca\sim\bern(p)^k$ and $\vecpi \sim \Unif(\sym_k)$, the distributions of $\vecpi(\vecb) \odot \veca$ and $\vecpi(\vecb \odot \veca)$ are identical. Thus, we have
\begin{align*}
    \Exp_{\veca \sim \bern(p)^n}[\Th^i_k(\vecb(i) \odot \veca|_{\vecj})] &= \Exp_{\veca \sim \bern(p)^k}[\Th^i_k(\vecb \odot \veca)] \tag{first observation}\\
    &= \Exp_{\veca \sim \bern(p)^k,\vecpi\sim\Unif(\sym_k)}[\Th^i_k(\vecpi(\vecb \odot \veca))] \tag{symmetry of $\Th^i_k$} \\
    &= \Exp_{\veca \sim \bern(p)^k,\vecpi\sim\Unif(\sym_k)}[\Th^i_k(\vecpi(\vecb) \odot \veca)] \tag{second observation}.
\end{align*}
Thus, by linearity of expectation, we have $\Exp_{\veca \sim \bern(p)^n}[\val_\Psi(\veca)] = \Exp_{\veca \sim \bern(p)^k,\vecb\sim\dsym_\Psi}[\Th^i_k(\vecb \odot \veca)] = \lambda_S(\dsym_\Psi,p)$, as desired.

For \cref{item:dsym-mu}, we have
\begin{align*}
    \frac1{kW}\sum_{i=1}^n \diff_i(\Psi) &= \frac1{kW} \sum_{\ell=1}^m w_\ell \sum_{t=1}^k b(\ell)_t \tag{definition of $\diff_i(\Psi)$} \\
    &= \frac1{W} \sum_{\ell=1}^m w_\ell \Exp_{\vecpi\sim\sym_k}[\pi(\vecb)_1] \tag{where $\vecpi(\vecb) = (\pi(\vecb)_1,\ldots,\pi(\vecb)_k)$} \\
    &= \Exp_{\vecb \sim \dsym_\Psi}[b_1] \tag{definition of $\dsym_\Psi$}\\
    &= \mu(\dsym_\Psi) \tag{definition of $\mu$}.
\end{align*}

Finally, \cref{item:dsym-bias} follows immediately from \cref{item:dsym-mu} and the definition $\bias(\Psi) = \frac1{kW} \sum_{i=1}^n |\diff_i(\Psi)|$.
\end{proof}

Now, we are equipped to prove the lemmas:

\begin{proof}[Proof of \cref{lemma:thresh-value-ub}]
Let $\vecopt \in \{-1,1\}^n$ denote the optimal assignment for $\Psi$. Then
\begin{align*}
    \val_\Psi &= \val_\Psi(\vecopt) \tag{definition of $\vecopt$} \\
    &= \val_{\Psi^\vecopt}(1^n) \tag{\cref{item:flip-val-ass} of \cref{prop:flips}} \\
    &= \lambda_S(\dsym_{\Psi^{\vecopt}}, 1) \tag{\cref{item:dsym-val} of \cref{prop:dsym} with $p=1$} \\
    &= \gamma_S(\dsym_{\Psi^{\vecopt}}) \tag{definition of $\gamma_S$, \cref{eqn:beta_S-gamma_S-def}} \\
    &\leq \gamma_{S,k}(\mu(\dsym_{\Psi^{\vecopt}})) \tag{definition of $\gamma_{S,k}$, \cref{eqn:beta_Sk-gamma_Sk-def}}\\
    &\leq \gamma_{S,k}(\bias(\Psi^\vecopt)) \tag{\cref{item:dsym-mu} of \cref{prop:dsym} and monotonicity of $\gamma_{S,k}$} \\
    &= \gamma_{S,k}(\bias(\Psi)) \tag{\cref{item:flip-bias} of \cref{prop:flips}},
\end{align*}
as desired.
\end{proof}

\begin{proof}[Proof of \cref{lemma:thresh-value-lb}]
Let $\vecmaj \in \{-1,1\}^n$ denote the assignment assigning $x_i$ to $1$ if $\diff_i(\Psi) \geq 0$ and $-1$ otherwise. Now
\begin{align*}
    \val_\Psi &= \val_{\Psi^\vecmaj} \tag{\cref{item:flip-val} of \cref{prop:flips}} \\
    &\geq \sup_{p \in [0,1]} \left(\Exp_{\veca \sim \bern(p)^n}[\val_{\Psi^\vecmaj}(\veca)]\right) \tag{probabilistic method} \\
    &= \sup_{p \in [0,1]} (\lambda_S(\dsym_{\Psi^\vecmaj},p)) \tag{\cref{item:dsym-val} of \cref{prop:dsym}} \\
    & \geq \beta_S(\dsym_{\Psi^\vecmaj}) \tag{definition of $\beta_S$, \cref{eqn:beta_S-gamma_S-def}} \\
    &\geq \beta_{S,k}(\mu(\dsym_{\Psi^\vecmaj})) \tag{definition of $\beta_{S,k}$, \cref{eqn:beta_Sk-gamma_Sk-def}} \\
    &= \beta_{S,k}(\bias(\Psi^\vecmaj)) \tag{\cref{item:dsym-bias} of \cref{prop:dsym}} \\
    &= \beta_{S,k}(\bias(\Psi)) \tag{\cref{item:flip-bias} of \cref{prop:flips}},
\end{align*}
as desired.
\end{proof}

Finally, we state another consequence of \cref{lemma:thresh-value-ub} --- a simple randomized, $O(n)$-time-and-space streaming algorithm for \emph{outputting} approximately-optimal assignments when the max-min method applies.

\begin{theorem}\label{thm:thresh-bias-output-alg}
Let $\Th^i_k$ be a threshold function and $p^* \in [0,1]$ be such that the max-min method applies, i.e., \[ \alpha(\Th^i_k) = \inf_{\cD_N \in \Delta_k} \left(\frac{\lambda_S(\cD_N,p^*)}{\gamma_{S,k}(\mu(\cD_N))}\right). \] Then the following algorithm, on input $\Psi$, outputs an assignment with expected value at least $\alpha(\Th^i_k) \cdot \val_\Psi$: Assign every variable to $1$ if $\diff_i(\Psi) \geq 0$, and $-1$ otherwise, and then flip each variable's assignment independently with probability $p^*$.
\end{theorem}

\begin{proof}
Let $p^*$ be as in the theorem statement, and define $\vecmaj$ as in the proof of \cref{lemma:thresh-value-lb}. We output the assignment $\vecmaj \odot \veca$ for $\veca \sim \bern(p^*)^n$, and our goal is to show that its expected value is at least $\alpha(\Th^i_k) \val_\Psi$.

Our assumption that the max-min method applies asserts in particular that
\begin{equation}\label{eqn:max-min-for-alg}
    \lambda_S(\dsym_{\Psi^\vecmaj},p^*) \geq \alpha(\Th^i_k) \gamma_{S,k}(\mu(\dsym_{\Psi^\vecmaj})).
\end{equation}
Thus our expected output value is
\begin{align*}
    \Exp_{\veca\sim\bern(p^*)}[\val_\Psi(\vecmaj \odot \veca)]  &= \Exp_{\veca\sim\bern(p^*)}[\val_{\Psi^\vecmaj}(\veca)] \tag{\cref{item:flip-val-ass} of \cref{prop:flips}}\\
    &=\lambda_S(\dsym_{\Psi^\vecmaj},p^*) \tag{\cref{item:dsym-val} of \cref{prop:dsym}} \\
    &\geq \alpha(\Th^i_k)\gamma_{S,k}(\mu(\dsym_{\Psi^\vecmaj})) \tag{\cref{eqn:max-min-for-alg}} \\
    &= \alpha(\Th^i_k) \gamma_{S,k}(\bias(\Psi)) \tag{\cref{item:dsym-bias} of \cref{prop:dsym}} \\
    &\geq \alpha(\Th^i_k) \val_\Psi \tag{\cref{lemma:thresh-value-ub}},
\end{align*}
as desired.
\end{proof}

\section*{Discussion}

In this paper, we introduce the max-min method and use it to resolve the streaming approximability of a wide variety of symmetric Boolean CSPs (including infinite families such as $\maxkand$ for all $k$, and $\Th_k^{k-1}$ for all even $k$). However, these techniques are in a sense ``ad hoc'' since we use computer assistance to guess the optimal solution for our optimization problem. We leave the question of whether the max-min method can be applied to determine the sketching approximability for all symmetric Boolean CSPs as an interesting open problem. 

Separately, we also establish that the techniques developed in \cite{CGSV21-boolean} are not sufficient to characterize the \emph{streaming} approximability of all CSPs. Indeed, we show that their streaming lower bound based on ``padded one-wise pairs'' cannot match the approximation ratio of their optimal sketching algorithm for \m[\threeand]. While we believe that no $o(\sqrt{n})$-space streaming algorithm can beat their sketching algorithm for \m[\threeand], proving this will require new techniques.
\section*{Acknowledgements}

All authors of this paper are or were supervised by Madhu Sudan. We would like to thank him for advice and guidance on both technical and organizational aspects of this project, as well as for many helpful comments on earlier drafts of this paper.
\newpage
\printbibliography

\newpage
\appendix
\section{Miscellaneous technical proofs}\label{sec:misc-proofs}

\begin{proof}[Proof of \cref{prop:lin-opt}]
Firstly, we show that it suffices WLOG to take the special case where $r=n$ and $\vecy(1),\ldots,\vecy(n)$ is the standard basis for $\R^n$. Indeed, assume the special case and note that for a general case, we can let $\veca'=(\veca\cdot\vecy(1),\ldots,\veca\cdot\vecy(r))$, $\vecb'=(\vecb\cdot\vecy(1),\ldots,\vecb\cdot\vecy(r))$, $\vecx'=(x_1,\ldots,x_r)$, and let $\vecy'(1),\ldots,\vecy'(r)$ be the standard basis for $\R^r$. Then $\vecx' = \sum_{i=1}^r \alpha_i \vecy'(i)$ and \[ f(\vecx) = \frac{\sum_{i=1}^r (\veca \cdot \vecy(i)) \alpha_i}{\sum_{i=1}^r (\vecb \cdot \vecy(i)) \alpha_i} = \frac{\veca' \cdot \vecx'}{\vecb' \cdot \vecx'} \geq \min_{i \in [r]} \frac{\veca' \cdot \vecy'(i)}{\vecb' \cdot \vecy'(i)} = \min_{i\in[r]} \frac{\veca \cdot \vecy(i)}{\vecb \cdot \vecy(i)}. \]

Now we prove the special case: Assume $r=n$ and $\vecy(1),\ldots,\vecy(n)$ is the standard basis for $\R^n$. We have $f(\vecy(i)) = \frac{a_i}{b_i}$. Assume WLOG that $f(\vecy(1)) = \min \{f(\vecy(i)): i \in [n]\}$, i.e., $\frac{a_1}{b_1} \leq \frac{a_i}{b_i}$ for all $i \in [n]$. Then $a_i \geq \frac{a_1b_i}{b_1}$ for all $i \in [n]$, so \[ \veca \cdot \vecx \geq \sum_{i=1}^n \frac{a_1b_i}{b_1} \alpha_i = \frac{a_1}{b_1} (\vecb \cdot \vecx). \] Hence \[ f(\vecx) = \frac{\veca \cdot \vecx}{\vecb \cdot \vecx} \geq \frac{a_1}{b_1} = f(\vecy(1)), \] as desired.
\end{proof}

\begin{proof}[Proof of \cref{thm:thresh-bias-alg}]
To get an $(\alpha-\epsilon)$-approximation to $\val_\Psi$, let $\delta > 0$ be small enough such that $\frac{1-\delta}{1+\delta}\alpha(\Th^i_k) \geq \alpha(\Th^i_k)-\epsilon$. We claim that calculating an estimate $\hat{b}$ for $\bias(\Psi)$ (using \cref{cor:bias-sketching}) up to a multiplicative $\delta$ factor and outputting $\hat{v} = \alpha(\Th^i_k)\gamma_{S,k}(\frac{\hat{b}}{1+\delta})$ is sufficient.

Indeed, suppose $\hat{b} \in [(1-\delta)\bias(\Psi),(1+\delta)\bias(\Psi)]$; then $\frac{\hat{b}}{1+\delta} \in [\frac{1-\delta}{1+\delta}\bias(\Psi),\bias(\Psi)]$. Now we observe
\begin{align*}
    \gamma_{S,k}\left(\frac{\hat{b}}{1+\delta}\right) &\geq \gamma_{S,k}\left(\frac{1-\delta}{1+\delta} \bias(\Psi)\right) \tag{monotonicity of $\gamma_{S,k}$} \\ 
    &=\min\left\{\frac{1+\frac{1-\delta}{1+\delta}\bias(\Psi)}{1+\epsilon_{s,k}},1\right\} \tag{\cref{lemma:gamma-formula}} \\
    &\geq \frac{1-\delta}{1+\delta} \min\left\{\frac{1+\bias(\Psi)}{1+\epsilon_{s,k}},1\right\} \tag{$\delta>0$} \\
    &= \frac{1-\delta}{1+\delta} \gamma_{S,k}(\bias(\Psi)) \tag{\cref{lemma:gamma-formula}}.
\end{align*}
Then we conclude
\begin{align*}
    (\alpha(\Th^i_k)-\epsilon)\val_\Psi &\leq (\alpha(\Th^i_k)-\epsilon)\gamma_{S,k}(\bias(\Psi)) \tag{\cref{lemma:thresh-value-ub}} \\
    &\leq \alpha(\Th^i_k) \cdot \frac{1-\delta}{1+\delta}\gamma_{S,k}(\bias(\Psi)) \tag{assumption on $\delta$} \\
    &\leq \hat{v} \tag{our observation} \\
    &\leq \alpha(\Th^i_k)\gamma_{S,k}(\bias(\Psi)) \tag{monotonicity of $\gamma_{S,k}$} \\
    &\leq \beta_{S,k}(\bias(\Psi)) \tag{\cref{eqn:alpha-optimize-over-mu}} \\
    &\leq \val_\Psi \tag{\cref{lemma:thresh-value-lb}} ,
\end{align*} as desired.
\end{proof}

\end{document}